\documentclass[conference,letterpaper]{IEEEtran}
\usepackage{ifthen}
\usepackage[bookmarks]{hyperref}
\usepackage{amssymb}
\usepackage[cmex10]{amsmath}
\usepackage{amsthm}
\usepackage[abs]{overpic}
\usepackage{tcolorbox}
\usepackage{tikz}
\usepackage{graphicx}
\usetikzlibrary{positioning}
\usepackage{color,colordvi}
\usepackage{bbm}
\usepackage{cleveref}
\usepackage{stmaryrd}
\usepackage[utf8]{inputenc}
\usepackage{fullpage}
\usepackage{dsfont}
\usepackage{mathtools}

\usepackage{centernot}
\usepackage{pgfplots}

\def\Re{{\operatorname{Re}}}

\def\openone{\leavevmode\hbox{\small1\kern-3.8pt\normalsize1}}

\def\cn{{\cal N}}

\def\CC{\mathbb{C}}

\def\11{\mathbb{I}}

\def\sln{\succeq_{\operatorname{l.n.}}}
\def\mc{\succeq_{\operatorname{m.c.}}}
\def\deg{\succeq_{\operatorname{deg}}}
\def\tdeg{\succeq_{\operatorname{t-deg}}}

\newtheorem{definition}{Definition}[section]
\newtheorem{proposition}[definition]{Proposition}
\newtheorem{lemma}[definition]{Lemma}

\newtheorem{corollary}[definition]{Corollary}

\newcommand{\tr}{\mathop{\rm Tr}\nolimits}

\usepackage{pst-node}
\usepackage{tikz-cd} 

\newcommand{\bra}[1]{\langle#1|}
\newcommand{\ket}[1]{|#1\rangle}

\newcommand{\cA}{{\cal A}}
\newcommand{\cB}{{\cal B}}
\newcommand{\cC}{{\cal C}}
\newcommand{\cD}{{\cal D}}
\newcommand{\cE}{{\cal E}}
\newcommand{\cF}{{\mathcal{F}}}

\newcommand{\cN}{{\cal N}}

\newcommand{\cH}{{\cal H}}

\newcommand{\cR}{{\cal R}}

\newcommand{\cM}{{\mathcal{M}}}
\newcommand{\cO}{{\cal O}}

\def\d{\mathrm{d}}

\newcommand{\notimplies}{\;\not\!\!\!\implies}

\def\sln{\succeq_{\operatorname{l.n.}}}

\usepackage{graphicx}
\usepackage{setspace}
\usepackage{verbatim}
\usepackage{subfig}

\numberwithin{equation}{section}

\DeclareRobustCommand\openone{\leavevmode\hbox{\small1\normalsize\kern-.33em1}}

\newcommand{\id}{{\rm{id}}}
\newcommand{\be}{\begin{equation}}
	\newcommand{\ee}{\end{equation}}
\newcommand{\bea}{\begin{eqnarray}}
	\newcommand{\eea}{\end{eqnarray}}
\newcommand{\beas}{\begin{eqnarray*}}
	\newcommand{\eeas}{\end{eqnarray*}}

\DeclareFontFamily{U}{mathx}{\hyphenchar\font45}
\DeclareFontShape{U}{mathx}{m}{n}{<-> mathx10}{}
\DeclareSymbolFont{mathx}{U}{mathx}{m}{n}
\DeclareMathAccent{\widebar}{0}{mathx}{"73}

\newcommand{\Hell}{H}

\renewcommand{\d}{\textnormal{d}}

\renewcommand{\Re}{D}
\newcommand{\Id}{{\mathds{1}}}
\DeclareMathAccent{\widehat}{0}{mathx}{"70}
\DeclareMathAccent{\widecheck}{0}{mathx}{"71}

\begin{document}
\title{Quantum Doeblin coefficients: A simple upper bound on contraction coefficients}

\author{%
  \IEEEauthorblockN{Christoph Hirche}
  \IEEEauthorblockA{Institute for Information Processing (tnt/L3S), Leibniz Universit\"at Hannover, Germany}

}


\maketitle

\begin{abstract}
Contraction coefficients give a quantitative strengthening of the data processing inequality. As such, they have many natural applications whenever closer analysis of information processing is required. However, it is often challenging to calculate these coefficients. As a remedy we discuss a quantum generalization of Doeblin coefficients. These give an efficiently computable upper bound on many contraction coefficients. We prove several properties and discuss generalizations and applications. In particular, we give additional stronger bounds. One especially for PPT channels and one for general channels based on a constraint relaxation. Additionally, we introduce reverse Doeblin coefficients that bound certain expansion coefficients. 
\end{abstract}

\section{Introduction}

Classically, Doeblin first introduced what is now known as the Doeblin minorization condition~\cite{doeblin1937proprietes}, which holds if for a channel $P_{Y|X}$ there exists a constant $\alpha\in[0,1]$ and a probability distribution $Q_Y$, such that 
\begin{align}
    P_{Y|X}(y|x) \geq \alpha\, Q_Y(y), \quad \forall x,y. \label{Eq:C-Doeblin-min}
\end{align}
The Doeblin coefficient of $P_{Y|X}$ is simply the largest such constant $\alpha$. Its optimal value can be expressed as 
\begin{align}
    \alpha(P_{Y|X}) = \sum_y \min_x P_{Y|X}(y|x), \label{Eq:C-Doeblin-coe}
\end{align}
which is how it most commonly appears in the literature. This coefficient has found plenty of applications, for example in change detection~\cite{chen2022change}, multi-armed bandits~\cite{moulos2020finite}, Markov chain Monte Carlo methods~\cite{rosenthal1995minorization}, Markov decision processes~\cite{alden1992rolling} and mixing-models~\cite{steinhardt2015learning}. 
Quantum generalizations of the above concept have so far gained little attention, although a quantum version of Equation~\eqref{Eq:C-Doeblin-min} was discussed in~\cite{wolf2012quantum}. Here we want to introduce quantum Doeblin coefficients via a slightly different route. In the classical setting it is known that $\alpha(P_{Y|X})$ can be equivalently expressed as
\begin{align}
    \alpha(P_{Y|X})= \sup\{\epsilon : \cE_\epsilon \deg P_{Y|X} \},
\end{align}
where $\cE_\epsilon$ denotes the erasure channel with erasure probability $\epsilon$ and $\cE_\epsilon \deg P_{Y|X}$ means that $P_{Y|X}$ is degraded with respect to $\cE_\epsilon$. We refer to the next section for formal definitions. Motivated by the above, we define the quantum Doeblin coefficient of a quantum channel $\cN$ as
\begin{align}
    \alpha(\cN):= \sup\{\epsilon : \cE_\epsilon \deg\cn \},
\end{align}
where $\cE_\epsilon$ is now the quantum erasure channel. 
In this work we will discuss properties and applications of this quantity. In particular, we will see that it gives an efficiently computable upper bound on a wide class of contraction coefficients. These bounds take the form of semidefinite programs (SDP). 
Contraction coefficients give strengthenings of the data processing inequality, which is itself of vital importance to the theory of quantum information processing. Unfortunately, they are notoriously hard to evaluate numerically. Still they have found many applications. This makes good bounds very valuable. 

Motivated by the results, we will also expand the concept of Doeblin coefficients in several ways. First, by replacing degradability with transpose degradability we get an alternative bound for positive partial transpose (PPT) channels. This \textit{transpose Doeblin coefficient} sometimes performs better than the original version and leads to substantially improved bounds. 
Second, we give an improved bound by showing that one of the SDP constraints can be relaxed while still giving a valid bound. 
Third, we define \textit{reverse Doeblin coefficients} which, instead of upper bounding contraction coefficients, give lower bounds on expansion coefficients. This applies in particular to the expansion coefficient of the trace distance. 
The particular feature of these reverse Doeblin coefficients is that the erasure channel gets replaced by the depolarizing channel. 

In the next section, we will start by giving the necessary formal definitions. In Section~\ref{Sec:Doeblin-coef}, we discuss Doeblin coefficients, their properties and the transpose variant. In Section~\ref{Sec:rev-Doeblin} we introduce and discuss the reverse Doeblin coefficients. Section~\ref{Sec:Applications} will discuss some further general extensions and applications. Finally, we will conclude in Section~\ref{Sec:Conclusions}. 

\section{Background and definitions}

We start, by giving the necessary technical definitions. 
Here, we consider $f$-divergences as defined in~\cite{hirche2023quantum} for $f\in\cF:=\{f\; \text{convex}, f(1)=0  \}$ as,
\begin{align}
    &D_f(\rho\|\sigma) \nonumber \\
    &:= \int_1^\infty f''(\gamma) E_\gamma(\rho\|\sigma) + \gamma^{-3} f''(\gamma^{-1})E_\gamma(\sigma\|\rho)\, \d \gamma ,
\end{align}
where $E_\gamma(\rho\|\sigma) = \tr(\rho-\gamma\sigma)_+$ is the quantum hockey stick divergence. For $\gamma=1$, this is the trace distance $E_1(\rho\|\sigma)=\frac12\|\rho-\sigma\|_1$. We are particularly interested in the contraction of $f$-divergences under quantum channels.
Contraction coefficients give a strong version of data processing inequalities by quantifying the decrease in distinguishability.
We define the contraction coefficient for the $f$-divergence for a quantum channel $\cA$ as
\begin{align}
\eta_f(\cA) :=  \sup_{\rho,\sigma} \frac{D_f(\cA(\rho) \| \cA(\sigma))}{D_f(\rho\|\sigma)}. 
\end{align}
Clearly we have $0\leq \eta_f(\cA) \leq 1$. By definition the contraction coefficients give us the best channel dependent constants such that
\begin{align}
    D_f(\cA(\rho) \| \cA(\sigma)) \leq \eta_f(\cA) D_f(\rho\|\sigma) \qquad\forall\rho,\sigma.
\end{align}
An important special case is that of the relative entropy, $D(\rho\|\sigma)=\tr\rho(\log\rho-\log\sigma)$. We define,  
\begin{align}
    \eta_\Re(\cA) := \eta_{x\log x}(\cA) = \sup_{\rho,\sigma} \frac{D(\cA(\rho) \| \cA(\sigma))}{D(\rho\|\sigma)}. 
\end{align}

Besides the relative entropy we are also interested in the contraction of the trace distance,
\begin{align}
    \eta_{\tr}(\cA) = \sup_{\rho,\sigma} \frac{E_1(\cA(\rho) \| \cA(\sigma))}{E_1(\rho\|\sigma)},
\end{align}
which is also often called the Dobrushin coefficient. While contraction coefficients can be hard to compute, $\eta_{\tr}(\cA)$ has a comparably simpler form~\cite{ruskai1994beyond}, 
\begin{align}
    \eta_{\tr}(\cA) = \sup_{\ket{\Psi}\perp\ket{\Phi}} E_1(\cA(\ket{\Psi}\bra{\Psi}) \| \cA(\ket{\Phi}\bra{\Phi})),
\end{align}
where the optimization is over two orthogonal pure states. Nevertheless, even in this case, it is not known whether the quantity can be efficiently computed. For a broader overview of the topic of contraction coefficients we refer to~\cite{hirche2022contraction,hiai2016contraction,hirche2023quantum}.
Some results proven in~\cite{hirche2023quantum} will be useful later. For example, we have the inequality, 
 \begin{align}
         \eta_f(\cA) \leq \eta_{\tr}(\cA), 
 \end{align}
and for every two operator convex functions $f,g$, we have
    \begin{align}
        \eta_f(\cA) &= \eta_{g}(\cA). \label{Eq:f-x2-contraction-eq}
    \end{align}
This includes in particular the relative entropy. 

Furthermore, we call a channel $\cM$ less noisy than a channel $\cN$, denoted $\cM \sln \cN$, if
\begin{align}
    I(U:B)_{\cM(\rho)} \geq I(U:B')_{\cN(\rho)}\qquad\forall \rho_{UA}, \label{Eq:less-noisy-MI}
\end{align}
where $I(A:B)_{\rho_{AB}}=D(\rho_{AB}\|\rho_A\otimes \rho_B)$ is the quantum mutual information and $\rho_{UB}=(\id_U\otimes\cM)(\rho_{UA})$, $\rho_{UB'}=(\id_U\otimes\cN)(\rho_{UA})$ with $\rho_{UA}$ a classical-quantum state of the form 
\begin{align}
    \rho_{UA} = \sum_u p(u) |u\rangle\langle u| \otimes \rho^u_A. 
\end{align}

For any two operator convex functions $f$ and $g$, we have~\cite{hirche2023quantum}, 
\begin{align}
    \eta_f(\cN) = \sup_{\rho_{UA}} \frac{I_g(U:B)}{I_g(U:A)}, \label{Eq:etaf-MI-OC}
\end{align}
where $I_g(A:B)_{\rho_{AB}}=D_g(\rho_{AB}\|\rho_A\otimes \rho_B)$ and the supremum is over all classical-quantum states $\rho_{UA}$. 

The quantum erasure channel $\cE_{\epsilon}$ with erasure probability $\epsilon$, where $\cH_{B}=\cH_A\oplus \CC$, is defined for any $X\in\cB(\cH_A)$ as,
\begin{align}
\cE_{\epsilon}(X):=(1-\epsilon)X+\epsilon|e\rangle\langle e|\,,
\end{align}
where $|e\rangle\perp\cH_A$. The above concepts are further connected via the following observation. 

\begin{proposition}[Proposition II.5,~\cite{hirche2022contraction}]\label{Prop:erasureContraction}
	Let $\cN$ be a quantum channel and $\epsilon\in[0,1]$. Then the following are equivalent:
	\begin{itemize}
	\item[(i)] $\cE_{\epsilon}\sln\cn$.
		\item[(ii)] $\eta_{D}(\cn)\le \eta_{D}(\cE_{\epsilon})=(1-\epsilon)$.
	\end{itemize}
\end{proposition}

As a consequence we can write for any operator convex $f$ with $f(1)=0$, 
\begin{align}
    \eta_f(\cN) = 1 - \beta(\cN), 
\end{align}
with 
\begin{align}
    \beta(\cN):= \sup\{\epsilon : \cE_\epsilon \sln\cn \}. 
\end{align}
Here, we are interested in efficiently computable upper bounds on these contraction coefficients. 
Another partial order is that of degradability. We call a channel $\cN$ a degraded version of $\cM$, denoted $\cM \deg \cN$, if there exists another channel $\cD$, such that $\cN=\cD\circ\cM$. Clearly, as a consequence of data processing, if a channel is a degraded version of another, then it also more noisy with respect to the previous partial order. In the following, $\cM\geq\cN$ means $\cM-\cN$ is CP and whenever we compare a state $\sigma$ with a channel, we really mean $\sigma$ as the replacer channel $\cR_\sigma(\cdot) = \sigma \tr(\cdot)$. By $J(\cN)$ we denote the normalized Choi operator of $\cN$, i.e. $J(\cN)=\cN\otimes\id(\Psi)$, where $\Psi$ is the maximally entangled state. 

\section{Quantum Doeblin coefficients} \label{Sec:Doeblin-coef}

We will now investigate a quantum generalization of Doeblin coefficients. As mentioned in the introduction, we define the quantum Doeblin coefficient as
\begin{align}
    \alpha(\cN):= \sup\{\epsilon : \cE_\epsilon \deg\cn \}.
\end{align}
From the previous discussion, we can now easily see that this coefficient gives an upper bound on a wide class of contraction coefficients.
\begin{corollary}\label{Cor:alpha-beta}
    For a quantum channel $\cN$, we have, 
    \begin{align}
        \beta(\cN) \geq \alpha(\cN), 
    \end{align}
    and hence, for any operator convex $f$ with $f(1)=0$, 
    \begin{align}
    \eta_f(\cN) \leq 1 - \alpha(\cN). \label{Eq:Bound-c-f}
\end{align}
\end{corollary}
\begin{proof}
    The first inequality follows because degradability implies less noisy via data processing. The second inequality is implied by the results from~\cite{hirche2022contraction,hirche2023quantum} discussed above: We have,  
    \begin{align}
         \eta_f(\cN) = \eta_D(\cN) = 1-\beta(\cN) \leq 1 - \alpha(\cN),
    \end{align}
    which proves the claim.
\end{proof} 
First, we give a strengthening of Equation~\eqref{Eq:Bound-c-f}. We will later see that the following result is equivalent to~\cite[Theorem 8.17]{wolf2012quantum} after giving an alternative expression for the quantum Doeblin coefficient. For completeness, we will give a simple alternative proof below. 
\begin{lemma}\label{Lem:Ineq-trace-alpha}
    We have, 
    \begin{align}
    \eta_{\tr}(\cN) \leq 1 - \alpha(\cN). \label{Eq:Bound-c-tr}
\end{align}
\end{lemma}
\begin{proof}
    Let $\epsilon^\star$ be the optimizer in $\alpha(\cN)$, then
    \begin{align}
          &\eta_{\tr}(\cN) \\
          &= \sup_{\ket{\Psi}\perp\ket{\Phi}} E_1(\cN(\ket{\Psi}\bra{\Psi}) \| \cN(\ket{\Phi}\bra{\Phi})) \\
          &= \sup_{\ket{\Psi}\perp\ket{\Phi}} E_1(\cD\circ\cE_{\epsilon^\star}(\ket{\Psi}\bra{\Psi}) \| \cD\circ\cE_{\epsilon^\star}(\ket{\Phi}\bra{\Phi})) \\
          &\leq \sup_{\ket{\Psi}\perp\ket{\Phi}} E_1(\cE_{\epsilon^\star}(\ket{\Psi}\bra{\Psi}) \| \cE_{\epsilon^\star}(\ket{\Phi}\bra{\Phi})) \\
          &= \sup_{\ket{\Psi}\perp\ket{\Phi}} (1-\epsilon^\star)E_1(\ket{\Psi}\bra{\Psi} \| \ket{\Phi}\bra{\Phi})  \\
          &= 1-\epsilon^\star \\
          &= 1-\alpha(\cN), 
    \end{align}
    where the first two and the last two equalities are by definition, the inequality is data processing and the third equality follows from the structure of the erasure channel.  
\end{proof}
Next, we will give alternative expressions for $\alpha(\cN)$ providing the bridge to the classical definition via Equation~\eqref{Eq:C-Doeblin-min} and the quantum Doeblin minorization introduced in~\cite{wolf2012quantum}. 
\begin{proposition}\label{Prop:alpha-expressions}
    We have 
    \begin{align}
        \alpha(\cN) &= \max\{\epsilon : \exists\cD,\sigma \;\text{s.t.}\; \cN = (1-\epsilon)\cD + \epsilon \sigma  \} \\
        &= \max\{c\in[0,1] : \exists\sigma \;\text{s.t.}\; c\,\sigma \leq \cN  \}. \label{Eq:Q-Doeblin-min}
    \end{align}
\end{proposition}
\begin{proof}
    We start by proving the first equality. We see directly that 
    \begin{align}
        \cD\circ\cE_\epsilon = (1-\epsilon) \cD + \epsilon D(|\epsilon\rangle\langle\epsilon|). 
    \end{align}
    Choosing $\sigma=\cD(|\epsilon\rangle\langle\epsilon|)$ we get the $\leq$ direction of our statement. Similarly we can construct a $\cD'$ such that $\cD'(\rho)=\cD(\rho)$ and $\cD'(|\epsilon\rangle\langle\epsilon|)=\sigma$ which is always possible because the erasure state is orthogonal to the input space. This proves the $\geq$ direction and concludes the argument. 
    Now, we prove the second equality. 
    We start with a standard argument, 
    \begin{align}
        &c \sigma \leq \cN \label{Eq:order-1}\\
        &\Leftrightarrow \cN - c\sigma \in CP \\
        &\Leftrightarrow \exists\cD \;\text{s.t.}\; \cN - c\sigma =(1-c)\cD \\
        &\Leftrightarrow \exists\cD \;\text{s.t.}\; \cN =(1-c)\cD + c\sigma\label{Eq:order-4}
    \end{align}
    which essentially follows from the definition of the CP order on channels. This implies, 
    \begin{align}
        &\max\{\epsilon : \exists\cD,\sigma \;\text{s.t.}\; \cN = (1-\epsilon)\cD + \epsilon \sigma  \} \\
        &= \max\{c\in[0,1] : \exists\sigma \;\text{s.t.}\; c\,\sigma \leq \cN  \},
    \end{align}
   and concludes the proof.  
\end{proof}
The above expressions might remind the reader of the measures of robustness typically used in quantum resource theories. 
The expression in Equation~\eqref{Eq:Q-Doeblin-min} should now be compared to the classical Doeblin minorization condition in Equation~\eqref{Eq:C-Doeblin-min}. Hence, we establish the coefficient as the maximum value in the quantum Doeblin minorization introduced in~\cite{wolf2012quantum}. 
We also remark that following the above we can write 
\begin{align}
    \alpha(\cN) = \exp{\left(- \inf_\sigma D_{\max}(\sigma\|\,\cN\,)\right)}
\end{align}
with 
\begin{align}
    D_{\max}(\cM\|\cN)= \log\min\{\lambda : \cM\leq\lambda\cN\}, 
\end{align}
as defined in~\cite[Definition 19]{diaz2018using}, see also~\cite[Remark 13]{wilde2020amortized}. From this formulation it also becomes evident that the quantum Doeblin coefficient is similar to a quantity defined in~\cite[Theorem 21]{mishra2023optimal}. That quantity found an application as lower bound on the minimum error probability of antidistinguishability. Both coincide for classical quantum channels. 

From the preceding discussion it is now easy to see that we can express $\alpha(\cN)$ as a semidefinite program (SDP) which can be efficiently computed~\cite{boyd2004convex}. 
\begin{corollary}\label{Cor:alpha-SDP}
    Quantum Doeblin coefficients can be expressed as the following SDP, 
    \begin{align}
        \alpha(\cN) = \max_{\substack{
        \hat\sigma\geq 0 \\ \hat\sigma \otimes \frac{\Id}{d} \leq J(\cN) }} \tr\hat\sigma. 
    \end{align}
\end{corollary}
\begin{proof}
    We start by noting that to confirm complete positivity it is sufficient to check the Choi operators, 
    \begin{align}
        \alpha(\cN) &= \max\{c\in[0,1] : \exists\sigma \;\text{s.t.}\; c\,\sigma \leq \cN  \} \\
        &= \max\{c\in[0,1] : \exists\sigma \;\text{s.t.}\; c\,\sigma\otimes \frac{\Id}{d} \leq J(\cN)  \}.
    \end{align}
    We then introduce $\hat\sigma=c\sigma$ and note 
    \begin{align}
        0\leq c = \tr\hat\sigma = \tr\sigma\otimes \frac{\Id}{d} \leq \tr J(\cN) =1.
    \end{align}
    Therefore we can simplify the expression by eliminating $c$. This concludes the proof. 
\end{proof}

\subsection{Examples}
Knowing that we can compute the quantity, it is time for some examples. A common standard example is the depolarizing channel for $0\leq p\leq1$,
\begin{align}
    \cD_p(\cdot) = (1-p) \id(\cdot) + p\frac{\Id}{d}. 
\end{align}
Checking the expression in Equation~\eqref{Eq:Q-Doeblin-min} we can choose $\sigma=\frac\Id{d}$ and note that
\begin{align}
    c\frac\Id{d} \leq  (1-p) \id(\cdot) + p\frac{\Id}{d}
\end{align}
always holds for $c\leq p$, implying $\alpha(\cD_p)\geq p$. It is also intuitively clear that one cannot do better. We can also make this more formal, noting that by Lemma~\ref{Lem:Ineq-trace-alpha}, we have 
\begin{align}
    \alpha(\cD_p) \leq 1 - \eta_{\tr}(\cD_p) = p, 
\end{align}
where the final equality follows from a direct calculation. In summary, 
\begin{align}
    \alpha(\cD_p) = 1 - \eta_{\tr}(\cD_p) = p. 
\end{align}
We note however, that the depolarizing channel is more generally a valid quantum channel for all values $0\leq p\leq \frac{d^2}{d^2-1}$. For $1 < p\leq \frac{d^2}{d^2-1}$ we then have $\eta_{\tr}(\cD_p)=p-1$ and find numerically that in this range
\begin{align}
    \alpha(\cD_p) < 1 - \eta_{\tr}(\cD_p) = 2-p. 
\end{align}
In the following section we will present an improved bound that remedies this discrepancy. 

Next, consider the generalized amplitude damping channel, $\cA_{p,\eta}$, given by the Kraus operators~\cite{nielsen2000quantum},
\begin{align}
    A_1 &= \sqrt{p} \begin{bmatrix} 1 & 0 \\ 0 & \sqrt{\eta} \end{bmatrix}& A_2&=\sqrt{p}\begin{bmatrix}
        0 & \sqrt{1-\eta} \\ 0 & 0
    \end{bmatrix} \nonumber\\
    A_3&=\sqrt{1-p}\begin{bmatrix} \sqrt{\eta} & 0 \\ 0 & 1\end{bmatrix}& A_4&=\sqrt{1-p}\begin{bmatrix} 0 & 0 \\ \sqrt{1-\eta} & 0\end{bmatrix},
\end{align}
where $p\in[0,1]$ represents the dissipation to the environment and $\eta\in[0,1]$ is related to how much the input mixes with the environment. This is a somewhat more interesting example as it is less clear how to compute $\eta_{\tr}(\cA_{p,\eta})$.  The coefficient $\alpha(\cA_{p,\eta})$ is depicted in Figure~\ref{fig:GAD-alpha-1}. By virtue of the SDP formulation, generating numerical values is a matter of seconds on a standard laptop. We will revisit these and other examples later on. 

\begin{figure}
    \centering
     \begin{tikzpicture}
  \node (img)  { 
    \includegraphics[trim={4.25cm 10cm 4.25cm 10cm},clip,width=0.98\linewidth]{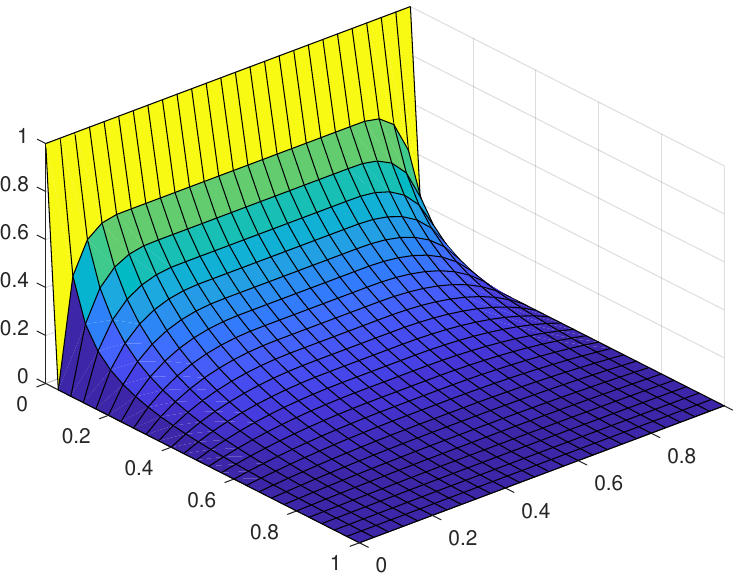}};
    \node[below=of img,  node distance=0cm, yshift=1.9cm,xshift=2.2cm] {$p$};
    \node[below=of img,  node distance=0cm, yshift=1.9cm,xshift=-2.2cm] {$\eta$};
 \end{tikzpicture}
    \caption{Plot of the Doeblin coefficient $\alpha(\cA_{p,\eta})$ for $p\in[0,1]$ and $\eta\in[0,1]$.}
    \label{fig:GAD-alpha-1}
\end{figure}

\subsection{Transpose quantum Doeblin coefficients}

Our definition of the quantum Doeblin coefficient is obviously linked to degradability. To get better bounds on contraction coefficients it will be useful to vary this approach a bit. The concept of transpose degradability was defined in~\cite{singh2022detecting}. We refer also to~\cite{bradler2010conjugate} for the related idea of conjugate degradability. Let $T$ be the transpose map, which is well known to be positive but not completely positive. We call a channel $\cN$ a transpose degraded version of $\cM$, denoted $\cM \tdeg \cN$, if there exists another channel $\cD$, such that 
\begin{align}
    T\circ\cN=\cD\circ\cM. \label{Eq:t-deg}
\end{align}
With this we can define the transpose quantum Doeblin coefficient,
\begin{align}
    \alpha^T(\cN) &:= \sup\{\epsilon : \cE_\epsilon \tdeg\cn \} \\
    &= \sup\{\epsilon : \cE_\epsilon \deg T\circ\cn \}. 
\end{align}
From the definition of transpose degradability in Equation~\eqref{Eq:t-deg}, we see that existence of a degrading quantum channel $\cD$ places an immediate restriction on the applicability of this coefficient. That is, $\cN$ has to be a PPT channel, i.e. a channel for which $T\circ\cN$ is completely positive. For convenience, we can set $\alpha^T(\cN)=-\infty$ whenever $\cN$ is not PPT. 

The above definition is motivated by the following lemma. 
\begin{lemma}\label{Lem:Ineq-trace-talpha}
    We have, 
    \begin{align}
    \eta_{\tr}(\cN) \leq 1 - \alpha^T(\cN). \label{Eq:Bound-ct-tr}
    \end{align}
\end{lemma}
\begin{proof}
The proof is essentially identical to that of Lemma~\ref{Lem:Ineq-trace-alpha}. It differs only in adding the invariance of the trace distance under transpose, leading to
    \begin{align}
          &\eta_{\tr}(\cN) \\
          &= \sup_{\ket{\Psi}\perp\ket{\Phi}} E_1(\cN(\ket{\Psi}\bra{\Psi}) \| \cN(\ket{\Phi}\bra{\Phi})) \\          
          &= \sup_{\ket{\Psi}\perp\ket{\Phi}} E_1(T\circ\cN(\ket{\Psi}\bra{\Psi}) \| T\circ\cN(\ket{\Phi}\bra{\Phi})) \\
          &= \sup_{\ket{\Psi}\perp\ket{\Phi}} E_1(\cD\circ\cE_{\epsilon^\star}(\ket{\Psi}\bra{\Psi}) \| \cD\circ\cE_{\epsilon^\star}(\ket{\Phi}\bra{\Phi})),
    \end{align}
    where $\epsilon^\star$ is the optimizer of $\alpha^T(\cN)$. The proof is then continued as in Lemma~\ref{Lem:Ineq-trace-alpha}.
\end{proof}
For a PPT channel $\cN$ this hence gives an improved bound,
    \begin{align}
        \eta_{\tr}(\cN) \leq 1 - \max\{\alpha(\cN),\alpha^T(\cN)\}. \label{Eq:Bound-max-tr}
    \end{align}
It is also easy to see that $\alpha^T(\cN)$ has an SDP formulation very similar to that of $\alpha(\cN)$, 
    \begin{align}
        \alpha^T(\cN) = \max_{\substack{
        \hat\sigma\geq 0 \\ \hat\sigma \otimes \frac{\Id}{d} \leq J(T\circ\cN) }} \tr\hat\sigma.  
    \end{align}

\begin{figure}
    \centering
    \begin{tikzpicture}
  \node (img)  { \includegraphics[width=0.97\linewidth]{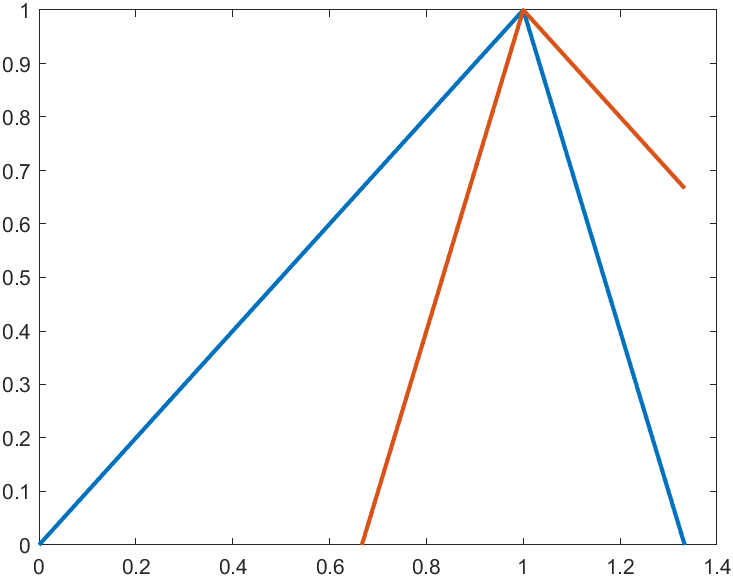}};
  \node[below=of img,  node distance=0cm, yshift=1cm] {$p$};
 \end{tikzpicture}
    \caption{Plot of the Doeblin coefficients $\alpha(\cD_p)$ (blue) and $\alpha^T(\cD_p)$ (red) for $p\in[0,\frac43]$. Recall that the qubit depolarizing channel is PPT for $p\geq\frac23$.}
    \label{fig:Dp-T}
\end{figure}

With this, we return to the previous example of the depolarizing channel. We plot $\alpha(\cD_p)$ and $\alpha^T(\cD_p)$ for the qubit depolarizing channel in Figure~\ref{fig:Dp-T}. This channel is completely positive for $p\in[0,\frac43]$ and PPT for $p\in[\frac23,\frac43]$. Evidently, for $p\geq1$ we get a better result from $\alpha^T(\cD_p)$ and we recover a tight bound on the contraction coefficient $\eta_{\tr}(\cD_p)$ via Equation~\eqref{Eq:Bound-max-tr}. 
Finally, we remark that a similar observation holds for the transpose-depolarizing channel, 
    \begin{align}
        \cD^T_q(\cdot) = (1-q) T(\cdot) + q\frac{\Id}{d}, 
    \end{align}
which is completely positive for $q\in[\frac{d}{d+1},\frac{d}{d-1}]$ and PPT for $q\in[\frac{d}{d+1},\frac{d^2}{d^2-1}]$. Here the roles of $\alpha$ and $\alpha^T$ are essentially interchanged. Note that $\cD^T_{\frac{d}{d-1}}$ is the often considered Werner-Holevo channel~\cite{werner2002counterexample}. 

\subsection{Tighter bound from a relaxed constraint}\label{Sec:hermitian}

Revisiting~\cite[Theorem 8.17]{wolf2012quantum}, we notice that the result there is actually formulated for a slightly more general setting in which all channels are only required to be hermiticity-preserving. While we are usually interested in completely positive maps, this interestingly still leads to a tighter bound. Building on Proposition~\ref{Prop:alpha-expressions}, we define
    \begin{align}
        \alpha^H(\cN) &:= \max\{c\in[0,1] : \exists X \;\text{s.t.}\; c\,X \leq \cN  \}\\ 
        &= \max\{\epsilon : \exists\cD,X \;\text{s.t.}\; \cN = (1-\epsilon)\cD + \epsilon X  \} . \label{Eq:Q-Doeblin-H}
    \end{align}
where $X$ is now any hermitian operator with trace one. In particular, we do not require it to be positive. A similar relaxiation has also proven useful in~\cite{mishra2023optimal}. The following is then a consequence of~\cite[Theorem 8.17]{wolf2012quantum}, but we also give a shorter proof.
\begin{lemma}\label{Lem:Ineq-trace-alpha-H}
    We have, 
    \begin{align}
    \eta_{\tr}(\cN) \leq 1 - \alpha^H(\cN). \label{Eq:Bound-c-tr-H}
\end{align}
\end{lemma}
\begin{proof}
    Let $\epsilon^\star$ be the optimizer in $\alpha^H(\cN)$, then
    \begin{align}
          &\eta_{\tr}(\cN) \\
          &= \sup_{\ket{\Psi}\perp\ket{\Phi}} E_1(\cN(\ket{\Psi}\bra{\Psi}) \| \cN(\ket{\Phi}\bra{\Phi})) \\
          &= \sup_{\ket{\Psi}\perp\ket{\Phi}} E_1((1-\epsilon^\star)\cD(\ket{\Psi}\bra{\Psi}) \| (1- \epsilon^\star) \cD(\ket{\Phi}\bra{\Phi})) \\
          &= (1-\epsilon^\star) \sup_{\ket{\Psi}\perp\ket{\Phi}} E_1(\cD(\ket{\Psi}\bra{\Psi}) \| \cD(\ket{\Phi}\bra{\Phi})) \\
          &\leq (1-\epsilon^\star) \sup_{\ket{\Psi}\perp\ket{\Phi}} E_1(\ket{\Psi}\bra{\Psi} \| \ket{\Phi}\bra{\Phi}) \\
          &= 1-\epsilon^\star \\
          &= 1-\alpha(\cN), 
    \end{align}
    where equalities are by definition or properties of the trace distance. The inequality is from data processing. 
\end{proof}
Clearly, we have, 
    \begin{align}
    \eta_{\tr}(\cN) \leq 1 - \alpha^H(\cN)\leq 1 - \alpha(\cN). 
\end{align}
The relaxation does not change the SDP property. To be precise, we have the following. 
\begin{corollary}\label{Cor:alpha-H-SDP}
    The relaxed quantum Doeblin coefficients can be expressed as the following SDP, 
    \begin{align}
        \alpha^H(\cN) = \max_{\substack{
        \hat X\;\text{hermitian} \\ \hat X \otimes \frac{\Id}{d} \leq J(\cN) }} \tr\hat X. 
    \end{align}
\end{corollary}
\begin{proof}
    The proof is almost identical to Corollary~\ref{Cor:alpha-SDP}. The only difference is that we don't automatically have $\tr\hat X\geq0$ in the claimed expression. However, we don't need to require that explicitly, because there always exists a feasible solution with $\tr\hat X = 0$. Hence, the optimal solution will not be smaller. 
\end{proof}

Revisiting our previous examples, we can also confirm that we get indeed a tangible improvement. In Figure~\ref{fig:alpha-H-comp}, we plot $1-\alpha(\cA_{p,\eta})$ and $1-\alpha^H(\cA_{p,\eta})$ for the generalized amplitude damping channel over $p$ for fixed values $\eta$. 
\begin{figure}
    \centering
        \begin{tikzpicture}
  \node (img)  {\includegraphics[width=0.95\linewidth]{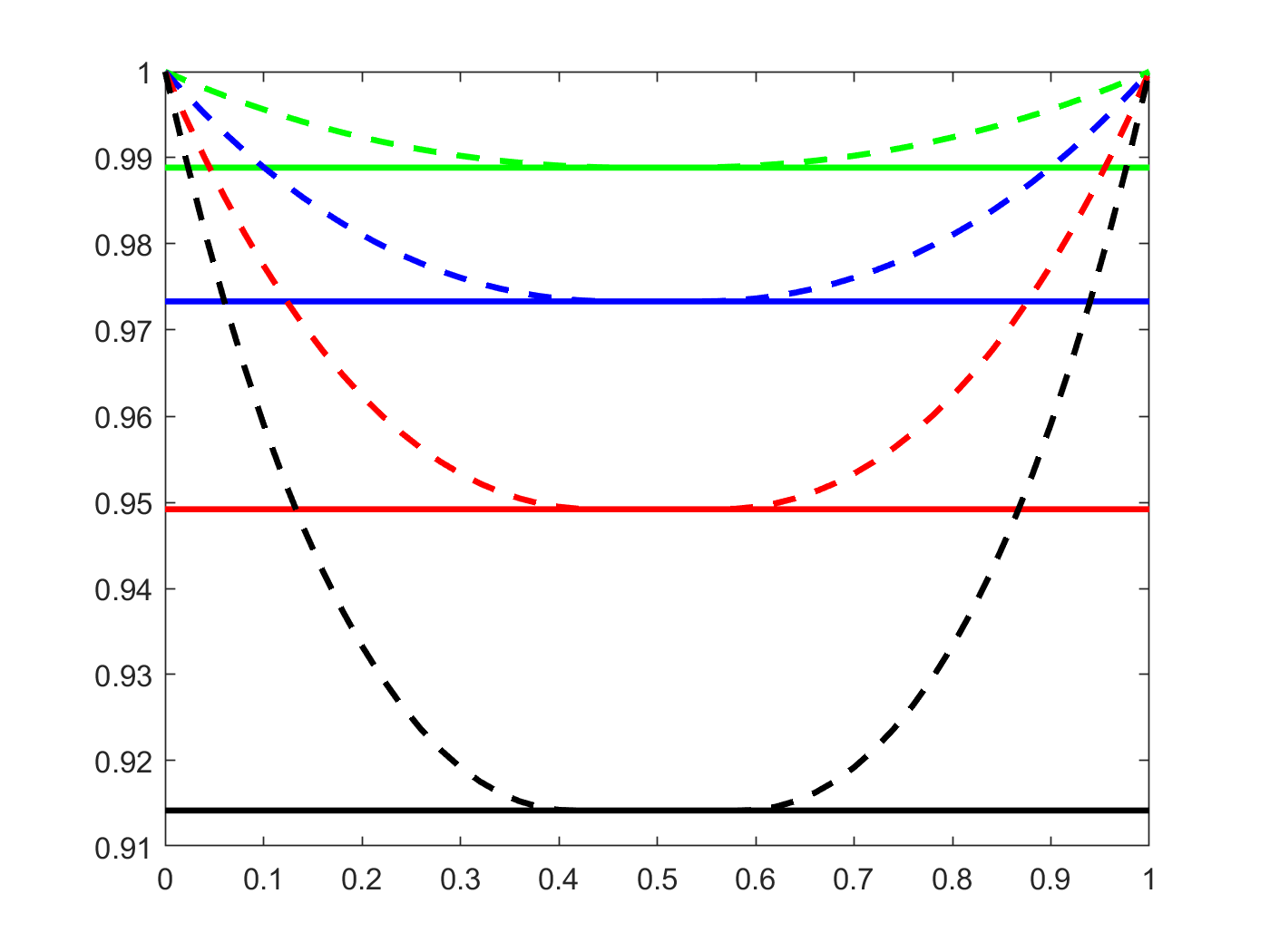}};
    \node[below=of img,  node distance=0cm, yshift=1.35cm] {$p$};
 \end{tikzpicture}
    \caption{Plot of $1-\alpha(\cA_{p,\eta})$ (dashed) and $1-\alpha^H(\cA_{p,\eta})$ (solid) over $p$ for $\eta=0.5$ (black), $0.6$ (red), $0.7$ (blue) and $0.8$ (green). } 
    \label{fig:alpha-H-comp}
\end{figure}
The difference is remarkable, in particular for values where $1-\alpha(\cA_{p,\eta})$ gives a trivial bound on the contraction coefficient. We find that $1-\alpha^H(\cA_{p,\eta})$ does never become trivial and hence detects correctly that $\eta_{\tr}(\cA_{p,\eta})<1$ for $\eta<1$. This can be crucial for many applications where we are primarily interested in whether the contraction coefficient is non-trivial. 

Finally, we remark that in some cases, e.g. the depolarizing channel, one still gets a better bound using $\alpha^T(\cN)$. Hence, the best upper bound presented in this work is
\begin{align}
     \eta_{\tr}(\cN) \leq 1 - \max\{\alpha^H(\cN),\alpha^T(\cN)\}. \label{Eq:Bound-max-tr-H}
\end{align}
One could also combine the transpose idea with the hermiticity relaxation. However, we do not know of any concrete examples where this gives a better bound. The most interesting remaining question is whether the bound in Equation~\eqref{Eq:Bound-max-tr-H} can reliably detect all cases of $\eta_{\tr}(\cN)<1$.

\subsection{Further properties}

In this section, we give a couple of properties of the quantum Doeblin coefficient that generalize those of its classical counterpart. 
\begin{lemma}\label{Lemma:Properties}
For quantum channels $\cN,\cM$, the quantum Doeblin coefficient $\alpha$ has the following properties: 
    \begin{enumerate}
        \item{(Concavity)} $\alpha(\cN)$ is concave in $\cN$, i.e. for $0\leq\lambda\leq1$, 
        \begin{align}
           \alpha(\lambda\cN + (1-\lambda)\cM) \geq \lambda \alpha(\cN) + (1-\lambda) \alpha(\cM) . 
        \end{align}
        \item{(Super-multiplicative)} We have, 
        \begin{align}
            \alpha(\cN\otimes\cM) \geq \alpha(\cN) \alpha(\cM). 
        \end{align}
        \item{(Concatenation)} We have, 
        \begin{align}
            1-\alpha(\cN\circ\cM) \leq (1-\alpha(\cN))(1-\alpha(\cM)). 
        \end{align}
    \end{enumerate}
\end{lemma}
\begin{proof}
 The proof can be found in Appendix~\ref{App:proofs}.
\end{proof}
We note that the second property above was previously shown for classical-quantum channels in~\cite{mishra2023optimal}. 

The above can be used to give bounds on contraction coefficients, e.g. 
\begin{align}
    \eta_{\tr}(\cN^{\otimes n}) \leq 1 - \alpha(\cN)^n.  
\end{align}
or
\begin{align}
    \eta_{\tr}(\cN\circ\dots\circ\cN) \leq (1 - \alpha(\cN))^n. 
\end{align}
That implies that even for combinations of many channels, it is numerically easy to give a bound using $\alpha(\cN)$ of the individual channels. 

\section{Reverse Doeblin coefficients}\label{Sec:rev-Doeblin}

Contraction coefficients bound how much contraction we will get at least from the application of a channel. Sometimes we might also be interested in a bound telling us that there won't be more contraction than a certain multiplicative constant. For that case, we can define expansion coefficients, 
\begin{align}
\widecheck\eta_f(\cA) :=  \inf_{\rho,\sigma} \frac{D_f(\cA(\rho) \| \cA(\sigma))}{D_f(\rho\|\sigma)},  
\end{align}
where we optimize over $\rho,\sigma$ such that $D_f(\rho\|\sigma)<\infty$. Similarly, we define $\widecheck\eta_f(\cA,\sigma)$ without the infimum over $\sigma$. For the quantum relative entropy such expansion coefficients were previously considered in~\cite{ramakrishnan2020computing}. Here, we will mainly look at the expansion coefficient for the trace distance, 
\begin{align}
\widecheck\eta_{\tr}(\cA) :=  \inf_{\rho,\sigma} \frac{E_1(\cA(\rho) \| \cA(\sigma))}{E_1(\rho\|\sigma)}.   
\end{align}
Naturally, we would like to apply our tools also to these coefficients. By virtue of data processing, the first idea would be to simply exchange the order of degradation in the Doeblin coefficient. However, in general it is not promising to find a channel $\cD$ such that $\cD\circ\cN=\cE_\epsilon$. Later, we will argue that in the classical setting it is sometimes reasonable to replace the erasure channel by a binary symmetric channel. Here, we propose for the quantum setting to use the depolarizing channel. Hence, we define the reverse Doeblin coefficient as 
\begin{align}
    \widecheck\alpha(\cN):= \inf\{p : \cN \deg\cD_p \}.
\end{align}
Clearly we have $\widecheck\alpha(\cN)\in[0,1]$ as $p\geq0$ is necessary for $\cD_p$ to be completely positive and every channel can be degraded into the fully depolarizing channel $\cD_1$. Most importantly, this coefficient lower bounds the expansion coefficient. 
\begin{lemma}\label{Lem:Ineq-trace-walpha}
    We have, 
    \begin{align}
    \widecheck\eta_{\tr}(\cN) \geq 1 - \widecheck\alpha(\cN). \label{Eq:Bound-c-tr}
\end{align}
\end{lemma}
\begin{proof}
Let $p^\star$ be the optimizer in $\widecheck\alpha(\cN)$ and $\cD$ the corresponding degrading map, then we have
\begin{align}
    E_1(\cN(\rho) \| \cN(\sigma)) &\geq E_1(\cD\circ\cN(\rho) \| \cD\circ\cN(\sigma)) \\
    &= E_1(\cD_{p^\star}(\rho) \| \cD_{p^\star}(\sigma)) \\
    &= (1-p^\star)\, E_1(\rho \| \sigma) \\
    &= (1-\widecheck\alpha(\cN))\, E_1(\rho \| \sigma). 
\end{align}
Hence the result follows from the definition of the expansion coefficient. 
\end{proof}
Next, we will see that also this quantity can be efficiently computed in the form of an SDP. Since, the degrading map is now applied to an arbitrary channel we do not have the same simple connection to the $D_{\max}$ relative entropy as before. We will instead use a different observation and introduce the link product from~\cite{chiribella2009theoretical},
\begin{align}
    M_{AB}\star N_{BC} := \tr_B \left[ (M_{AB}^{T_B}\otimes\Id_E)(\Id_A\otimes N_{BC})\right].
\end{align}
With that, we have, 
\begin{align}
    J(\cD\circ\cN)=J(\cD)\star J(\cN), 
\end{align}
which allows us to optimize over $J(\cD)$. This leads to the following SDP. 
\begin{corollary}\label{Cor:walpha-SDP}
    Reverse quantum Doeblin coefficients can be expressed as the following SDP, 
    \begin{align}
        \widecheck\alpha(\cN) = \min_{\substack{
        D\geq 0 \\ \tr_C D = \frac{\Id}{d_B} \\ J(\cN)\star D = (1-p)\Phi^+_{AC} + \frac{p\Id}{d_Ad_C} }} p,  
    \end{align}
    where $\Phi^+_{AC}$ is the maximally entangled state.
\end{corollary}
\begin{proof}
We have, 
\begin{align}
    \widecheck\alpha(\cN)
    &= \inf\{p : \cN \deg\cD_p \} \\
    &= \inf\{p : \exists \cD \;\text{s.t.}\; J(\cD\circ\cN) = J(\cD_p) \} \\
    &= \inf\{p : \exists \cD \;\text{s.t.}\; J(\cD)\star J(\cN) = J(\cD_p) \}. 
\end{align}
Casting $J(\cD)\equiv D$ as a variable in the optimization with the usual constraints on Choi states gives the desired SDP. 
\end{proof}
Motivated by our earlier discussion, we also define the \textit{reverse transpose quantum Doeblin coefficient}, 
\begin{align}
    \widecheck\alpha^T(\cN):= \inf\{p : \cN \deg\cD^T_p \}.
\end{align}
Interestingly, in this case we are not restricted to PPT channels $\cN$. Instead, $p$ needs to take values such that $D^T_p$ is completely positive. In particular, this implies that $\widecheck\alpha^T(\cN)\geq\frac{d}{d+1}$. While this seems like a major restriction, we will soon see examples for which $\widecheck\alpha^T$ gives better bounds then $\widecheck\alpha$. Before moving on, we briefly mention that, of course, $\widecheck\alpha^T$ can be formulated as an SDP, similar to Corollary~\ref{Cor:walpha-SDP}, and that it improves the bound on the expansion coefficient to
\begin{align}
    \widecheck\eta_{\tr}(\cN) \geq 1 - \min\{\widecheck\alpha(\cN),\widecheck\alpha^T(\cN)\}. \label{Eq:Bound-c-tr-max}
\end{align}
Both statements can be proven similarly to their respective analog results discussed previously. 
Next, we will discuss some examples in the following section. 

\subsection{Examples}
First, we revisit the depolarizing channel. For this channel, the expansion coefficient is easily checked to be equal to its contraction coefficient, 
\begin{align}
    \widecheck\eta_{\tr}(\cD_p) = \eta_{\tr}(\cD_p) = |1-p|. 
\end{align}
For $p\in[0,1]$ this bound is obviously achieved by the reverse Doeblin coefficient as it is itself based on the depolarizing channel. It is still interesting to look at $p\geq1$, because here we again get a better (optimal) bound from the transpose variant $\widecheck\alpha^T(\cD_p)$. We plot both coefficients in Figure~\ref{fig:Dp-Tw}. We remark that similar observations again also hold for the transpose depolarizing channel. 

\begin{figure}
    \centering
    \begin{tikzpicture}
  \node (img)  {\includegraphics[width=0.95\linewidth]{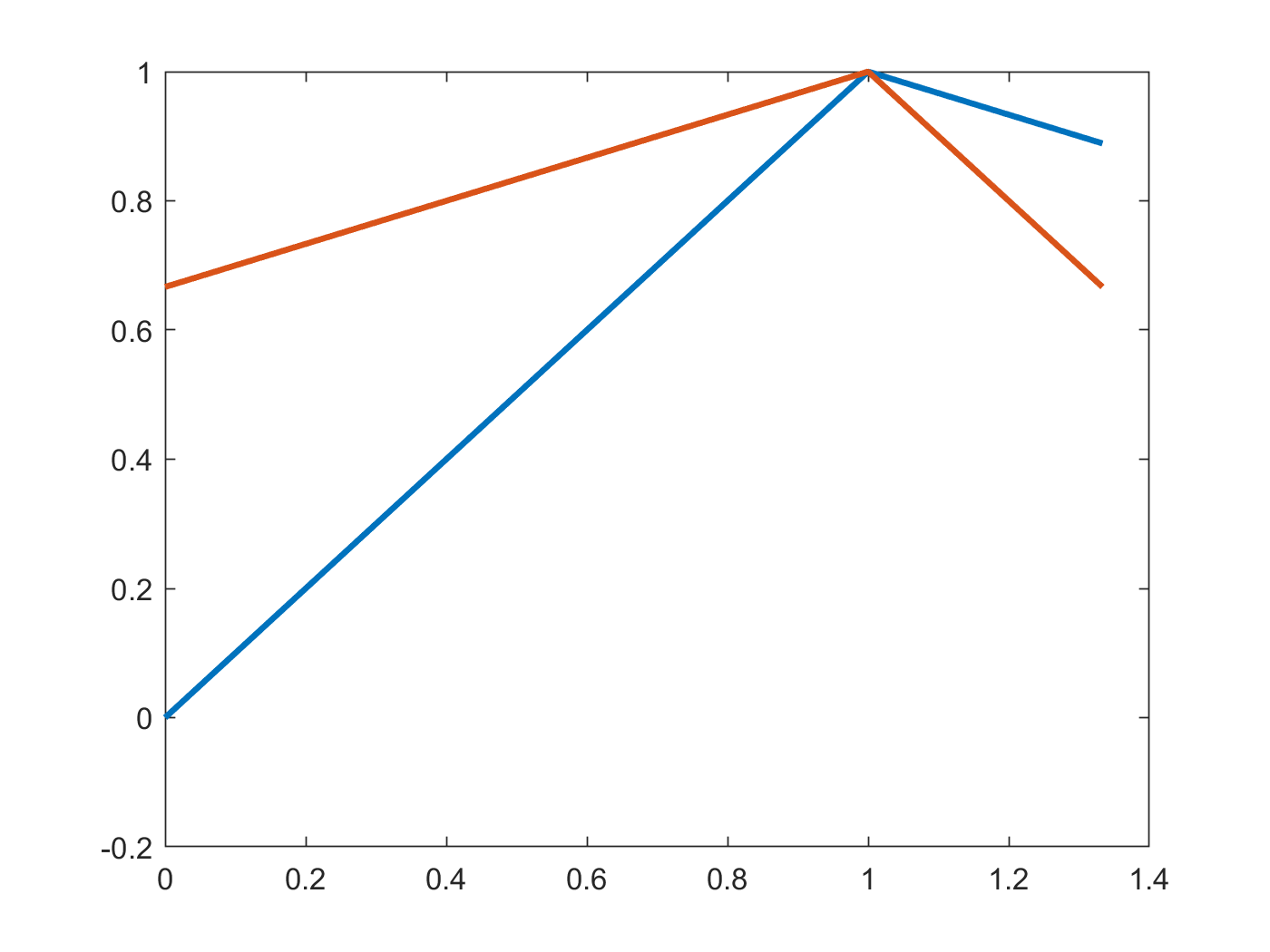}};
    \node[below=of img,  node distance=0cm, yshift=1.3cm] {$p$};
 \end{tikzpicture}
    \caption{Plot of the reverse Doeblin coefficients $\widecheck\alpha(\cD_p)$ (blue) and $\widecheck\alpha^T(\cD_p)$ (red) for $p\in[0,\frac43]$.}
    \label{fig:Dp-Tw}
\end{figure}

It is more interesting to see whether the reverse coefficient also gives a non-trivial, $\widecheck\alpha(\cN)<1$, bound on other channels. Again, we take the generalized amplitude damping channel and plot $\widecheck\alpha(\cA_{p,\eta})$ in Figure~\ref{fig:GAD-alpha-r}. One can clearly see that the reverse Doeblin coefficient places strong restrictions on the expansion coefficient. 

\begin{figure}
    \centering
    \begin{tikzpicture}
  \node (img)  {\includegraphics[width=0.95\linewidth]{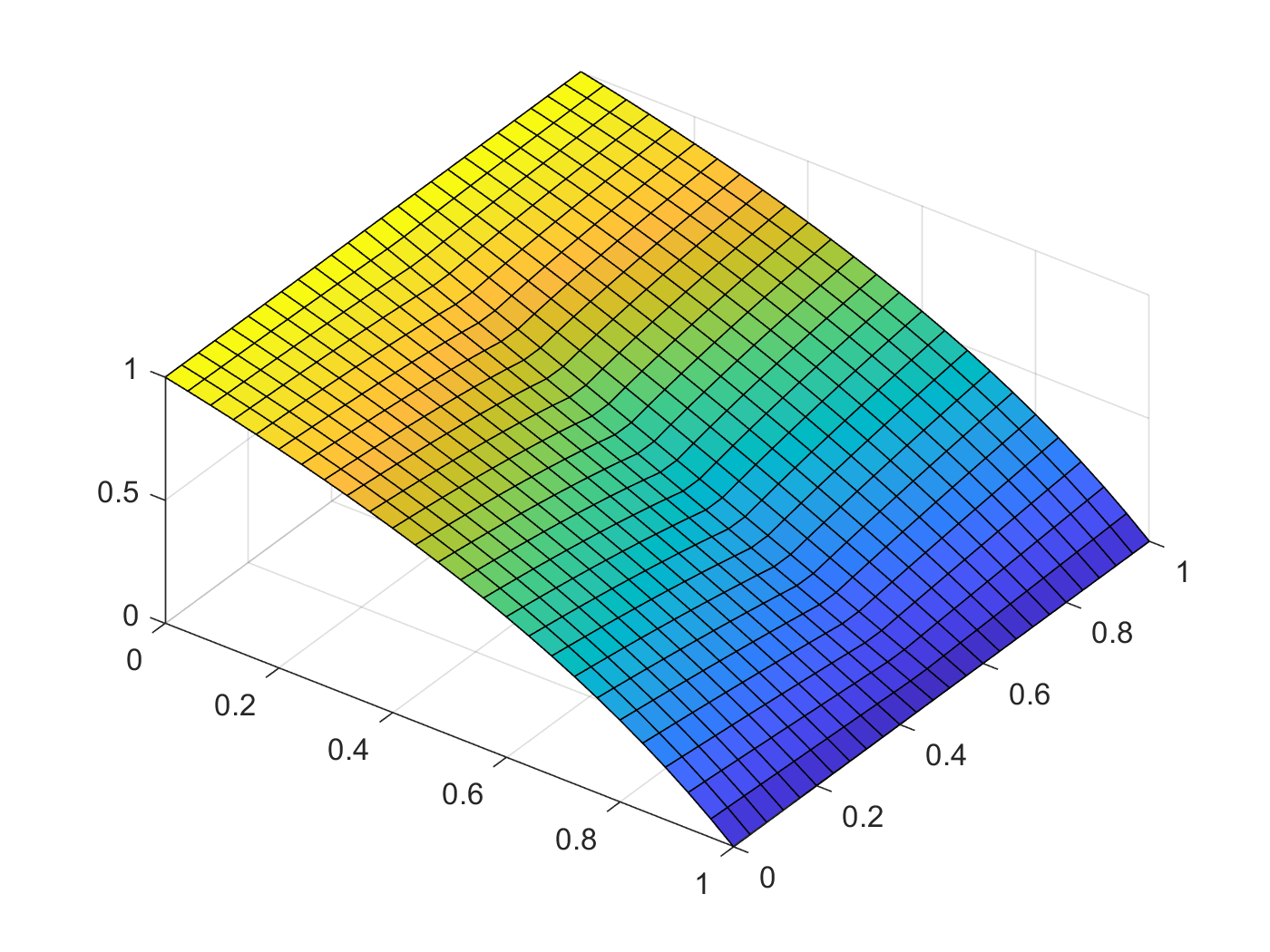}};
    \node[below=of img,  node distance=0cm, yshift=2cm,xshift=-2cm] {$\eta$};
  \node[below=of img,  node distance=0cm, yshift=2.2cm,xshift=2.4cm] {$p$};
 \end{tikzpicture}
    \caption{Plot of the reverse Doeblin coefficient $\widecheck\alpha(\cA_{p,\eta})$ for $p\in[0,1]$ and $\eta\in[0,1]$.}
    \label{fig:GAD-alpha-r}
\end{figure}

We can think of the gap between expansion and contraction coefficients as the $\textit{data processing range}$ of a channel $\cN$ with respect to a specified divergence. Doeblin coefficients then place lower and upper bounds on that range for the trace distance, 
\begin{align}
     1 - \widecheck\alpha(\cN) \leq \widecheck\eta_{\tr}(\cN)  \leq \eta_{\tr}(\cN) \leq 1 - \alpha(\cN), \label{Eq:Dp-range}
\end{align}
and analogously for the transpose variants. We visualize these bounds for the generalized amlpitude damping channel in Figure~\ref{fig:DP-Range}. 
Finally, consider the bit-flip channel, 
\begin{align}
    \cD^X_p(\cdot) = (1-p)\id(\cdot) + X(\cdot)X, 
\end{align}
where $X$ is the Pauli-X operator. It is well known that 
\begin{align}
    \eta_{\tr}(\cD^X_p)=1
\end{align}
for all $p\in[0,1]$. The expansion coefficient is more interesting. One can easily give an upper bound 
\begin{align}
    \widecheck\eta_{\tr}(\cD^X_p)\leq |1-2p|
\end{align}
by trying out the $|0\rangle$ and $|1\rangle$ states. The lower bound given by the reverse Doeblin coefficient leaves little room for the expansion coefficient. We show this in Figure~\ref{fig:X}.

\begin{figure}
    \centering
    \begin{tikzpicture}
  \node (img)  {\includegraphics[width=0.95\linewidth]{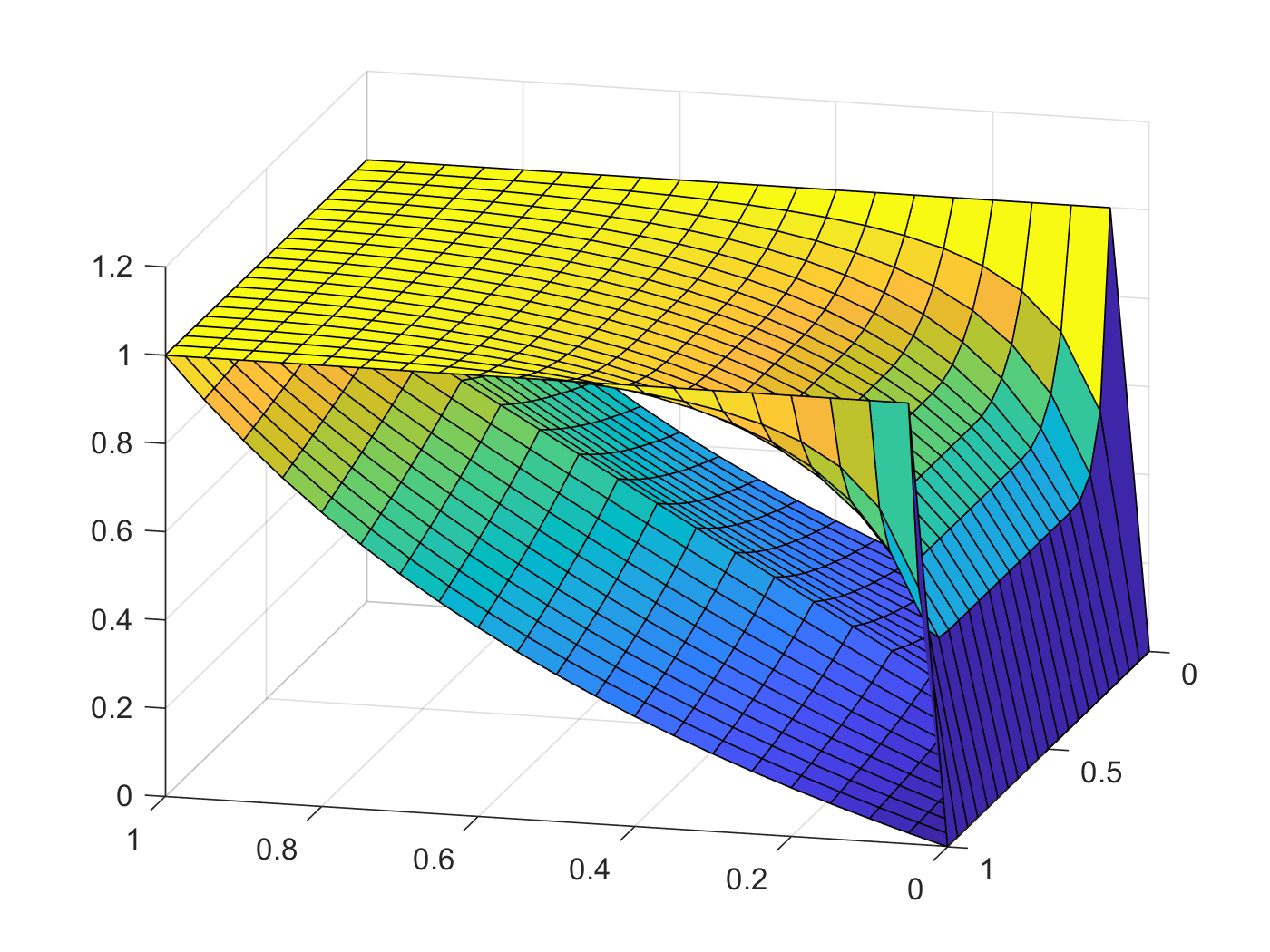}};
    \node[below=of img,  node distance=0cm, yshift=1.5cm,xshift=-1cm] {$\eta$};
  \node[below=of img,  node distance=0cm, yshift=2cm,xshift=3cm] {$p$};
 \end{tikzpicture}
    \caption{Plot of the bounds on the data processing range of $\cA_{p,\eta}$:  $1-\widecheck\alpha(\cA_{p,\eta})$ (lower bound) and $1-\alpha(\cA_{p,\eta})$ (upper bound) for $p\in[0,1]$ and $\eta\in[0,1]$.}
    \label{fig:DP-Range}
\end{figure}

\begin{figure}
    \centering
    \begin{tikzpicture}
  \node (img)  {\includegraphics[width=0.95\linewidth]{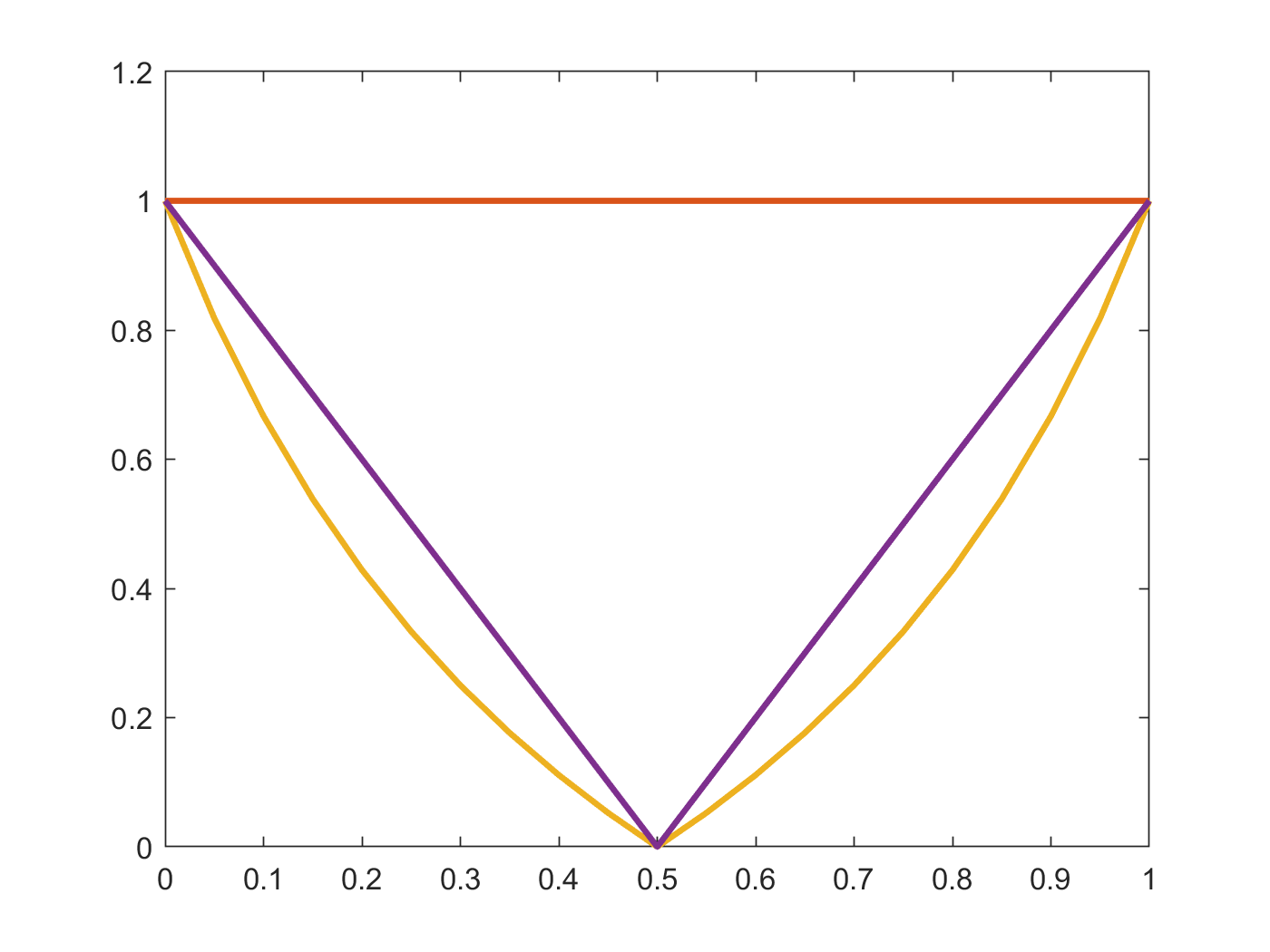}};
  \node[below=of img,  node distance=0cm, yshift=1.4cm] {$p$};
 \end{tikzpicture}
    \caption{Plot of the bounds on the data processing range of $\cD^X_p$:  $1-\widecheck\alpha(\cD^X_p)$ (yellow), $|1-2p|$ (blue) and $\eta_{\tr}(\cD^X_p)$ (red) for $p\in[0,1]$.}
    \label{fig:X}
\end{figure}

\subsection{Tighter bounds from generalized depolarizing channels}

In the above, we have defined the reverse Doeblin coefficient via the depolarizing channel. As mentioned before, this trivially gives an optimal contraction bound for the depolarizing channel itself. However, perhaps surprisingly, if we slightly change the channel that is no longer the case. Define, 
\begin{align}
    \cD_{p,\sigma}(\rho)=(1-p)\rho + p\sigma. 
\end{align}
Numerics suggests that $\alpha(\cD_{p,\sigma})$ does not reach the value $p$ necessary to optimally bound the trace distance contraction coefficient. This motivates us to give an improved bound where we allow for degradation into any channel of the form $\cD_{p,\sigma}$. In fact, we don't even need $\sigma$ to be a state and, similar to the hermitian relaxation above, allow for any $\cD_{p,X}$. Here, $X$ can be any hermitian operator with trace one. Define, 
\begin{align}
    \widecheck\alpha^H(\cN):= \inf\{p : \exists X\;\text{s.t.}\; \cN \deg\cD_{p,X} \}.
\end{align}
Alternatively, we can write this again as an SDP. 
\begin{corollary} 
The generalized reverse Doeblin coefficient can be expressed as the following SDP, 
\begin{align}
    \widecheck\alpha^H(\cN) = \min_{\substack{
        D\geq 0 \\
        \tilde X\;\text{hermitian} \\ \tr_C D = \frac{\Id}{d_B} \\ 
        J(\cN)\star D = (1-\tr\tilde X)\Phi^+_{AC} + \frac{\Id}{d_A}\otimes \tilde X }} \tr\tilde X,  
\end{align}
\end{corollary}
\begin{proof}
    By following the proof of Corollary~\ref{Cor:walpha-SDP}, we already get something very close to the claim. The final form follows by introducing $\tilde X=p X$ to make the constraint SDP compatible. 
\end{proof}

Following the ideas of the previous results, it is easy to see that 
\begin{align}
    \widecheck\eta_{\tr}(\cN) \geq 1- \widecheck\alpha^H(\cN). 
\end{align}
This is also clearly optimal for bounding $\widecheck\eta_{\tr}(\cD_{p,\sigma})$ for any $\sigma$ and hence gives a stronger bound. 
Possibly even more remarkable is that we also get a much stronger bound for the generalized amplitude damping channel. For illustrating purposes, we consider the standard amplitude damping channel. That is the GAD for $p=1$. 
In Figure~\ref{fig:AD-full}, the Doeblin coefficients are plotted each in their standard and tightened $H$ variants. Additionally we plot an inner bound for the data processing range of the trace distance.
This region is achieved by choosing example states to bound contraction and expansion coefficients. To be precise, the states $|+\rangle, |-\rangle$ give the upper border and the states $|1\rangle, |0\rangle$ give the lower border. 
We had already seen before, compare Figure~\ref{fig:alpha-H-comp}, that, while $\alpha(\cA_{1,\eta})$ is trivial, $\alpha^H(\cA_{1,\eta})$ gives a much more useful upper bound. 
Also the lower bound from $\widecheck\alpha(\cA_{1,\eta})$ gets much improved by switching to $\widecheck\alpha^H(\cA_{1,\eta})$. In fact, it appears that the resulting lower bound is optimal. Indeed, we confirm analytically in Appendix~\ref{App:GAD-exp}, that 
\begin{align}
        \widecheck\eta(\cA_{p,\eta}) = 1-\widecheck\alpha^H(\cA_{p,\eta}) = \eta. 
\end{align}

We make one final remark. In the definition of $\widecheck\alpha^H(\cA_{1,\eta})$, we chose to optimize over hermitian operators $X$, because it allows to optimize over a larger class than just positive operators. However, we note that for all examples presented, optimizing under the additional constraint that $X\geq0$ did not lead to any difference in the resulting values.

\begin{figure}
    \centering
    \begin{tikzpicture}
  \node (img)  {\includegraphics[width=0.975\linewidth]{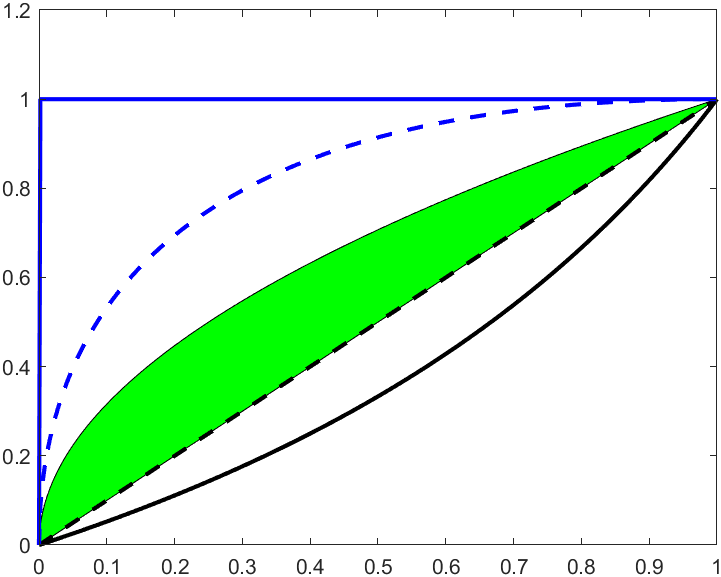}};
  \node[below=of img,  node distance=0cm, yshift=1cm] {$\eta$};
 \end{tikzpicture}
    \caption{Plot of the bounds on the data processing range of $\cA_{1,\eta}$:  $1-\alpha(\cA_{1,\eta})$ (blue, solid), $1-\alpha^H(\cA_{1,\eta})$ (blue, dashed), $1-\widecheck\alpha^H(\cA_{1,\eta})$ (black, dashed) and $1-\widecheck\alpha(\cA_{1,\eta})$ (black, solid) for $\eta\in[0,1]$. The green area gives the inner bound on the data processing range described in the main text. }
    \label{fig:AD-full}
\end{figure}

\subsection{Expansion for other divergences}
\label{Sec:Exp-f-div}
The Doeblin coefficient is by the previous discussion an upper bound on the contraction coefficient of every $f$-divergence. Naturally, we also wonder how different expansion coefficients are related to each other. We can easily find that the reverse Doeblin coefficient does not lower bound all expansion coefficients. For example, for the qubit depolarizing channel with $p\in[0,1]$~\cite{hiai2016contraction}, 
\begin{align}
    \eta_{D}(\cD_p) = (1-p)^2.
\end{align}
Hence, the reverse Doeblin coefficient doesn't even lower bound the larger contraction coefficient. We can in fact say a bit more about the expansion coefficients for $f$-divergences. First, we state the following lemma that is based on tools from~\cite{hirche2023quantum}. 
\begin{lemma}
    \label{Lem:x2-bound}
    For any four times differentiable $f\in\cF$ with $0<f''(1)<\infty$, we have
    \begin{align}
         \widecheck\eta_f(\cA,\sigma)&\leq \widecheck\eta_{x^2}(\cA,\sigma), \\
        \widecheck\eta_f(\cA)&\leq \widecheck\eta_{x^2}(\cA) ,
    \end{align}
    where $\eta_{x^2}$ is the contraction coefficient for the $\Hell_2$ divergence, also known as $\chi^2$ divergence, given by $f(x)=x^2-1$. 
\end{lemma}
\begin{proof}
    By definition we have,
    \begin{align}
        &D_f(\cA(\lambda\rho+(1-\lambda)\sigma)\|\cA(\sigma)) \nonumber\\
        &\quad\geq \widecheck\eta_f(\cA,\sigma) D_f(\lambda\rho+(1-\lambda)\sigma\|\sigma). \label{Eq:chi2-df-lb}
    \end{align}
    Now,~\cite[Theorem~2.8]{hirche2023quantum} essentially states that.  
    \begin{align}
        \lim_{\lambda\rightarrow 0} \frac{2}{\lambda^2} D_f(\lambda\rho+(1-\lambda)\sigma\|\sigma) = f''(1) \,\chi^2(\rho\|\sigma). 
    \end{align}
    Multiplying Equation~\eqref{Eq:chi2-df-lb} by $\frac2{\lambda^2}$ and taking the limit $\lambda\rightarrow 0$ gives then the first claim up to a factor $f''(1)$. Since we need to remove that factor, we get the condition $0<f''(1)<\infty$. The second claim follows by taking the infimum over all $\sigma$. 
\end{proof}
Returning once more to the qubit depolarizing channel, we can show that 
\begin{align}
    \widecheck\eta_{x^2}(\cD_p)=0. 
\end{align}
We give the details in Appendix~\ref{App:Ex-details}. 
That implies that all expansion coefficients that fulfill the condition of Lemma~\ref{Lem:x2-bound} are zero as well. This is, of course, not in contradiction to our previous results concerning the trace distance. Note that the trace distance is an $f$-divergence for $(t-1)_+$, which is not differentiable in $1$ where it hence also has a diverging second derivative. 
Given that the trace distance is a special case of the Hockey-Stick divergence, it seems natural to consider the latter next. The corresponding contraction coefficients have recently found an application in quantum differential privacy~\cite{hirche2022quantum}. However, it turns out that also these expansion coefficients are always zero for $\gamma>1$. Details can again be found in Appendix~\ref{App:Ex-details}. In summary, the trace distance expansion coefficient is so far the only one we found to be non-trivial for the qubit depolarizing channel. We also refer to~\cite{laracuente2023information} for related work. 

We leave it for future work to further explore expansion coefficients for general $f$-divergences and other channels. In particular, determining when they are non-zero seems like an interesting problem.

\section{Applications and Extensions}\label{Sec:Applications}

\subsection{Resource Theories}

Let us recall one of the expressions for $\alpha(\cN)$ and slightly rewrite it, 
\begin{align}
        &\alpha(\cN) \\
        &= \max\{c\in[0,1] : \exists\sigma \;\text{s.t.}\; c\,\sigma \leq \cN  \} \\
        &= \max\{c\in[0,1] : \exists\cR_\sigma \;\text{s.t.}\; c\,\cR_\sigma \leq \cN,\, \nonumber\\
        &\qquad\;\qquad\cR_\sigma\,\text{is a replacer channel}  \}. \label{Eq:Q-Doeblin-min-2}
\end{align}
Intuitively, we are hence comparing a given channel $\cN$ to the set of channels that delete all knowledge of the input state. Given that we are aiming to apply this to contraction coefficients, this is a natural choice as we always have $\eta_f(\cR_\sigma)=0$ since $D_f(\cR_\sigma(\rho)\|\cR_\sigma(\tau))=0$. 
While a simple observation, it gives us a path to think of applications with respect to more limited resources. The following quantity appeared in~\cite{fawzi2022lower} in the context of bounding the space overhead of fault-tolerant quantum computation,
\begin{align}
        p_1(\cN) &= \max\{c\in[0,1] : \exists\cA \;\text{s.t.}\; c\,\cA \leq \cN,\, \nonumber\\
        &\qquad\;\qquad\cA\,\text{is entanglement breaking}  \}. \label{Eq:Q-Doeblin-EB}
\end{align}
Here, our channel is compared to a set of channels that all completely destroy entanglement. Clearly, replacer channels are a special case of entanglement breaking channels and hence,
\begin{align}
    p_1(\cN) \geq \alpha(\cN). 
\end{align}
In general, if we have a given resource theory with a set of free states $\cF$, then we can define resource breaking channels, denoted $\cB_\cF$, as all those channels $\cN$ for which 
\begin{align}
    \cN(\rho) \in\cF \quad\forall\rho.
\end{align}
Naturally, a resource theory also comes with free operations, denoted $\cO_\cF$, which are all channels $\cM$ for which
\begin{align}
    \cM(\rho) \in\cF \quad\forall\rho\in\cF.
\end{align}
Denote by $\cN \leq_\cF \cM$ the partial order equivalent to $\cM-\cN\in\cO_\cF$. 
Define for a free operation $\cN$, 
\begin{align}
        &\alpha^\cF(\cN) \\
        &= \max\{c\in[0,1] : \exists\cR\in\cB_\cF \;\text{s.t.}\; c\,\cR \leq_\cF \cN \}. \label{Eq:Q-Doeblin-min-R} \\
        &=\max\{\epsilon : \exists\cD\in\cO_\cF,\cR\in\cB_\cF \nonumber\\
        &\qquad\quad\qquad\text{s.t.}\; \cN = (1-\epsilon)\cD + \epsilon \cR  \},
\end{align}
where the second equality follows along the lines of Equations~\eqref{Eq:order-1}-\eqref{Eq:order-4}. 
Similar to the previous considerations we can use this quantity to bound data processing. Let $R$ be a convex resource measure that obeys data processing. Meaning, for all $\cN\in\cO_\cF$, we have,
\begin{align}
    R(\cN(\rho)) \leq R(\rho). 
\end{align}
For such measures it makes sense to define a contraction coefficient as 
\begin{align}
    \eta_R(\cN) = \sup_\rho \frac{R(\cN(\rho))}{R(\rho)}. 
\end{align}
We have the following observation.
\begin{corollary}\label{Cor:eta-R-bound}
    Let $R$ be a convex resource measure and $\cN\in\cO_\cF$. Then,
    \begin{align}
        \eta_R(\cN) \leq 1 - \alpha^\cF(\cN). 
    \end{align}
\end{corollary}
\begin{proof}
    Let $\epsilon$ be the optimizer in $\alpha^\cF(\cN)$, then
    \begin{align}
        R(\cN(\rho)) &= R((1-\epsilon)\cD(\rho) + \epsilon \cR(\rho)) \\
        &\leq (1-\epsilon)R(\cD(\rho)) + \epsilon R(\cR(\rho)) \\
        &\leq (1-\epsilon)R(\rho) \\
        &= (1-\alpha^\cF(\cN))R(\rho) , 
    \end{align}
    where the first and last equalities are by definition, the first inequality by convexity and the second by data processing. Dividing both sides by $R(\rho)$ leads to the claim. 
\end{proof}
Also this quantity can be expressed as a particular $\max$-relative entropy, 
\begin{align}
    \alpha^\cF(\cN) = \exp{\left(- \inf_{\cR\in\cB_\cF} D^\cF_{\max}(\cR\|\cN)\right)},
\end{align}
where 
\begin{align}
D^\cF_{\max}(\cR\|\cN) = \log \min\{c : \cR \leq_\cF c\,\cN \}. 
 \end{align}
This restricted $\max$-relative entropy extends the definition of the cone restricted $\max$-relative entropy in~\cite{george2022cone} to quantum channels.

\subsection{Beyond degradability and erasure channels}

As we have seen in the previous sections, comparing to the erasure channel with respect to either the less noisy or degradable partial order has interesting implications. We also swapped that concept around by introducing reverse Doeblin coefficients that compare to the depolarizing channel. Naturally, one wonders what other combinations of partial orders and channels could lead to interesting quantities.  

Probably the most prominent remaining example of a channel partial order is the more capable partial order. We say that $\cN$ is more capable than $\cM$, denoted $\cN\mc\cM$, if 
\begin{align}
    I(X:B_\cN) \geq I(X:B_\cM)
\end{align}
for all
\begin{align}
    \rho_{XA} = \sum_x p(x) |x\rangle\langle x| \otimes \Psi_x.  \label{Eq:cq-pure}
\end{align}
Note that the only difference with the less noisy order is that now the classical quantum states are restricted to ensembles of pure states. 
We define the coefficient 
\begin{align}
    \gamma(\cN):= \sup\{\epsilon : \cE_\epsilon \mc \cn \}.
\end{align}
Naturally, 
\begin{align}
    \gamma(\cN) \geq \beta(\cN) = 1 - \eta_f(\cN) \geq \alpha(\cN),
\end{align}
for operator convex $f$. 
This quantity also has a formulation as contraction coefficient,
\begin{align}
    1 - \gamma(\cN) = \sup_{\rho_{XA}} \frac{I(X:B)}{I(X:A)},
\end{align}
where the supremum is over all states of the form given by Equation~\eqref{Eq:cq-pure}. 
Whether this can also be expressed in terms of relative entropies, similar to Equation~\eqref{Eq:etaf-MI-OC}, remains an interesting open problem. 

Of course, we can also define reverse coefficients based on the above. This, however, brings more questions than answers. In particular concerning the applications of these coefficients. 

To shed light on this question we take a step back and consider the classical setting. In particular, let $\cC_C$ be the set of all binary-input symmetric-output (BISO) channels with capacity $C$ and let $\cB_p$ be the binary symmetric channel (BSC). We propose the following reverse coefficients, 
\begin{align}
    \widecheck\alpha(\cN) &:= \inf\{h(p) : \cN \deg \cB_p \}, \\
    \widecheck\beta(\cN) &:= \inf\{h(p) : \cN \sln \cB_p \}, \\
    \widecheck\gamma(\cN) &:= \inf\{h(p) : \cN \mc \cB_p \}, 
\end{align}
where $h(p)$ is the binary entropy, which we introduce for normalization. 
Again, we immediately have, 
\begin{align}
    \widecheck\alpha(\cN) \geq \widecheck\beta(\cN) \geq \widecheck\gamma(\cN). 
\end{align}
However, it might not seem immediately clear how to compare e.g. $\gamma(\cN)$ and $\widecheck\gamma(\cN)$. To that end we give the following result. 
\begin{proposition}\label{Prop:gamma}
    For $\cN\in\cC_C$, we have, 
    \begin{align}
        \widecheck\gamma(\cN) = \gamma(\cN) = 1-C. 
    \end{align}
\end{proposition}
\begin{proof}
    One direction is easy. The capacity of a classical channel is given by
    \begin{align}
        C(\cN) = \max_{p(x)} I(X:Y), 
    \end{align}
    with the well known special cases, 
    \begin{align}
        C(\cE_\epsilon) &= 1-\epsilon, \\
        C(\cB_p) &= 1 - h(p). 
    \end{align}
    Now, by definition of $\gamma(\cN)$, we have 
    \begin{align}
        C = \max_{p(x)} I(X:Y) \leq C(\cE_\epsilon) = 1- \gamma(\cN),
    \end{align}
    and hence
    \begin{align}
        \gamma(\cN) \leq 1-C. 
    \end{align}
    Similarly, we get 
    \begin{align}
        1-C \leq \widecheck\gamma(\cN). 
    \end{align}
    For the other direction, we employ results from~\cite{geng2013broadcast}. In particular, combining~\cite[Corollary 1\&2]{geng2013broadcast}, we have for $\cN\in\cC_C$,
    \begin{align}
        \cE_{1-C} \mc \cN \mc \cB_{h^{-1}(1-C)}. \label{Eq:Geng}
    \end{align}
    This implies immediately, 
    \begin{align}
        \gamma(\cN)\geq 1-C \geq \widecheck\gamma(\cN), 
    \end{align} 
    because by Equation~\eqref{Eq:Geng} $1-C$ is always a feasible solution in either optimization. This concludes the proof. 
\end{proof}
Combined with previous observations, this gives for any BISO channel $\cN$, 
\begin{align}
    &\widecheck\alpha(\cN) \geq \widecheck\beta(\cN) \geq \widecheck\gamma(\cN) \nonumber\\
    &= 1-C  \\ \nonumber
    &= \gamma(\cN) \geq \beta(\cN) = 1 - \eta_f(\cN) \geq \alpha(\cN). 
\end{align}
Naturally, an interesting question is how to extend the above observations to quantum channels. In the classical case the choice of the binary symmetric channel is natural in the light of the results in~\cite{geng2013broadcast}, which are themselves connected to the question of information combining bounds. Previous work on quantum information combining bounds suggests that this setting is more involved~\cite{hirche2018bounds} even in the simple setting of symmetric binary classical-quantum channels. At least, the role of the binary symmetric channel did inform the choice of the depolarizing channel in the fully quantum setting which appears to be a natural substitution.

\subsection{Bounds on information theoretic quantities}

Partial orders between channels have long been a useful tool to bound different quantities, most importantly capacities~\cite{watanabe2012private,sutter2017approximate,hirche2022bounding}. Also the quantum Doeblin coefficient can prove useful in this task. That is mostly because many quantities, while potentially hard to compute in general, often simplify a lot when considering the erasure channel. For example, the quantum capacity of the qubit erasure channel is known~\cite{bennett1997capacities}, 
\begin{align}
    Q(\cE_\epsilon) = \max\{0,1-2\epsilon\}. 
\end{align}
Implying, by a simple use of data processing, that for a general qubit channel $\cN$, we have
\begin{align}
    Q(\cN) \leq \max\{0,1-2\alpha(\cN)\}.
\end{align}
Similarly, the classical capacity and the two-way quantum capacity can be bounded by,
\begin{align}
    Q_2(\cN) &\leq 1-\alpha(\cN), \\
     C(\cN) &\leq 1-\alpha(\cN).
\end{align}
A similar observation holds for every quantity that obeys data processing and is easy to compute for the erasure channel.

\section{Conclusions}\label{Sec:Conclusions}

In this work we introduced quantum Doeblin coefficients as an optimization over the erasure parameter that is needed to degrade the associated  erasure channel into a given target channel. Similar to the classical case, these turned out to be equivalent to the optimal constant in the Doeblin criteria, which had received a quantum generalization in~\cite{wolf2012quantum}. As a consequence they are efficiently computable SDP upper bounds on the contraction coefficients of all $f$-divergences as defined in~\cite{hirche2023quantum}. In particular, they also upper bound the trace distance contraction coefficient, which itself bounds e.g. the relative entropy coefficient. Hence, the bound seems to cover most relevant divergences. 

Additionally, we found improved bounds. One,  making use of the transpose-degradable partial order. Another one, by using a relaxation from optimizing over positive to hermitian operators. Depending on the quantum channel, these can give significantly improved numerical bounds. Additionally, we introduced reverse Doeblin coefficients that efficiently lower bound the expansion coefficient for the trace distance. 

On the other hand, we are also leaving several open problems that should be interesting for future research. We briefly mention two of them here: First, while efficiently computable bounds on contraction coefficients are very useful, the most important property is often to show whether the coefficient is strictly smaller than 1. For now, it is not clear when our bound does reliably detect these cases. The standard quantum Doeblin coefficient does not have that property. However, our improved bound in Equation~\eqref{Eq:Bound-max-tr-H} might have it. We do currently not know of any example where this is not the case. In that vein, there might also be room for further improvements along the lines of the transpose degradability variant. 
Second, it would be interesting to further investigate the data-processing range of a given divergence, i.e. the interval between the values of the expansion and contraction coefficients. In particular, it might be useful to understand better when expansion coefficients are trivial. For now, we only know of non-trivial examples for the trace distance. In that case, at least, they are easy to find thanks to the efficiently computable reverse Doeblin coefficients.

\section*{Acknowledgments}
I would like to thank Ruben Ibarrondo for many helpful comments, in particular for suggesting to look at the transpose-depolarizing channel. I thank Mark M. Wilde for very helpful feedback on an earlier version of this manuscript. That lead in particular to the improved bounds in Section~\ref{Sec:hermitian}. I also thank Peixue Wu for interesting discussions and pointing me to~\cite{laracuente2023information}.

\bibliographystyle{ultimate}
\bibliography{lib}

\newpage
\onecolumn
\appendices
\section{Proof of Lemma~\ref{Lemma:Properties}}
\label{App:proofs}

\begin{proof}
    We start with proving (1). To show concavity let $\epsilon_\cN$ and $\epsilon_\cM$ be the optimizers of $\alpha(\cN)$ and $\alpha(\cM)$, respectively, and $\cD_\cN$ and $\cD_\cM$ the corresponding degrading maps. We now show that there exists a degrading map $\cD$ such that 
    \begin{align}
        \lambda\cN + (1-\lambda)\cM = \cD \circ \cE_\epsilon, 
    \end{align}
    with $\epsilon=\lambda \alpha(\cN) + (1-\lambda) \alpha(\cM)$. Observe, 
    \begin{align}
        &\lambda\cN + (1-\lambda)\cM \\
        &= \lambda \left[ (1-\epsilon_\cN)\cD_\cN + \epsilon_\cN \cD_\cN(e)\right] + (1-\lambda) \left[ (1-\epsilon_\cM)\cD_\cM + \epsilon_\cM \cD_\cM(e)\right] \\
        &= \left[\lambda (1-\epsilon_\cN)\cD_\cN + (1-\lambda) (1-\epsilon_\cM)\cD_\cM\right] + \left[\lambda\epsilon_\cN \cD_\cN(e) +  (1-\lambda)\epsilon_\cM \cD_\cM(e)\right]. 
    \end{align}
    Now, define the map $\cD$ via
    \begin{align}
        \cD(e) =& \frac{\lambda\epsilon_\cN}{\lambda\epsilon_\cN+(1-\lambda)\epsilon_\cM}\cD_\cN(e) \nonumber\\ 
        &+ \frac{(1-\lambda)\epsilon_\cM}{\lambda\epsilon_\cN+(1-\lambda)\epsilon_\cM}\cD_\cM(e), 
    \end{align}
    and for $\rho\perp e$, 
    \begin{align}
        \cD(\rho) =& \frac{\lambda (1-\epsilon_\cN)}{\lambda (1-\epsilon_\cN)+(1-\lambda) (1-\epsilon_\cM)}\cD_\cN(\rho) \nonumber\\
        &+ \frac{(1-\lambda) (1-\epsilon_\cM)}{\lambda (1-\epsilon_\cN)+(1-\lambda) (1-\epsilon_\cM)}\cD_\cM(\rho). 
    \end{align}
    This achieves our goal and hence proves concavity. \\
    Next we proof (2). Following the same notation as before, observe
    \begin{align}
        &\cN\otimes\cM \\
        &= \left[ (1-\epsilon_\cN)\cD_\cN + \epsilon_\cN \cD_\cM(e)\right] \otimes \left[ (1-\epsilon_\cM)\cD_\cM + \epsilon_\cM \cD_\cM(e)\right] \\
        &= (1-\epsilon_\cN)(1-\epsilon_\cM)\cD_\cN\otimes\cD_\cM + (1-\epsilon_\cN)\epsilon_\cM\cD_\cN\otimes \cD_\cM(e) \nonumber\\
        &\qquad + \epsilon_\cN (1-\epsilon_\cM)\cD_\cM(e)\otimes\cD_\cM + \epsilon_\cN\epsilon_\cM \cD_\cM(e)\otimes \cD_\cM(e) \\
        &= \cD \circ \cE_{\epsilon_\cN\epsilon_\cM}
    \end{align}
    where $\cD$ can be constructed from the above derivation. This implies $\alpha(\cN\otimes\cM)\geq \epsilon_\cN\epsilon_\cM$ and hence the claim. \\
    Finally, we prove property (3). Here, we have
    \begin{align}
        &\cN\circ\cM \\
        &= \left[ (1-\epsilon_\cN)\cD_\cN + \epsilon_\cN \cD_\cM(e)\right] \circ \left[ (1-\epsilon_\cM)\cD_\cM + \epsilon_\cM \cD_\cM(e)\right] \\
        &= (1-\epsilon_\cN)(1-\epsilon_\cM)\cD_\cN\circ\cD_\cM +  (1-\epsilon_\cN)\epsilon_\cM\cD_\cN \circ  \cD_\cM(e) + \epsilon_\cN \cD_\cM(e)  \\
        &= \cD \circ \cE_{1-(1-\epsilon_\cN)(1-\epsilon_\cM)}
    \end{align}
    where we used in the second equality that $\cD_\cN(e)$ acts as a replacer channel and for the final equality constructed a $\cD$ which clearly exists because only the first term depends on the input state. This concludes the proof. 
\end{proof}

\section{Examples of expansion coefficients}\label{App:Ex-details}

In this section, we want to expand on the examples given in Section~\ref{Sec:Exp-f-div} and show how we evaluated them. We are interested here in showing that certain expansion coefficients are zero, hence a corresponding upper bound is sufficient. 
\subsection{$\chi^2$-divergence and relative entropy}
To keep things simple we choose classical input states, 
\begin{align}
    \rho=|0\rangle\langle 0|, \qquad 
    \sigma=(1-\epsilon)|0\rangle\langle 0| + \epsilon |1\rangle\langle 1|. \label{Eq:rs-examples}
\end{align}
For these we can easily calculate, 
\begin{align}
    D_f(\rho\|\sigma) &= (1-\epsilon) f\left(\frac{1}{1-\epsilon}\right) +\epsilon f\left(\frac0\epsilon\right), \\
    D_f(\cD_p(\rho)\|\cD_p(\sigma)) &= ((1-p)(1-\epsilon)+\frac{p}2) f\left(\frac{1-\frac{p}{2}}{(1-p)(1-\epsilon)+\frac{p}2}\right) + ((1-p)\epsilon+\frac{p}2) f\left(\frac{\frac{p}{2}}{(1-p)\epsilon+\frac{p}2}\right). 
\end{align}
Dividing the second by the first line then gives an upper bound on the expansion coefficient. Here we make the example $f(x)=x^2-1$ explicit, 
\begin{align}
    D_{x^2}(\rho\|\sigma) &= (1-\epsilon)^{-1} -1, \\
    D_{x^2}(\cD_p(\rho)\|\cD_p(\sigma)) &= (1-\frac{p}{2})^2((1-p)(1-\epsilon)+\frac{p}2)^{-1} + (\frac{p}{2})^2 ((1-p)\epsilon+\frac{p}2)^{-1} -1. 
\end{align}
For $\epsilon\rightarrow 0$, both of the above are $0$, hence we consider their derivatives, leading after some calculation to
\begin{align}
     \frac{\partial}{\partial \epsilon} D_{x^2}(\rho\|\sigma) \Big|_{\epsilon = 0} &= 1, \\
     \frac{\partial}{\partial \epsilon} D_{x^2}(\cD_p(\rho)\|\cD_p(\sigma)) \Big|_{\epsilon = 0} &= 0.
\end{align}
With this, by the rule of L'Hôpital, we find
\begin{align}
    \widecheck\eta_{x^2}(\cD_p) \leq 0  \qquad\text{for}\quad \epsilon\rightarrow0. 
\end{align}
A similar derivation works for $f(x)=x\log(x)$ as well, implying that the same states achieve the minimum in the limit. We remark that the relative entropy case, along with the mutual information and other examples of quantum channels, was also investigated in~\cite{laracuente2023information}. 

\subsection{Hockey-Stick divergence}

Given how similar the Hockey-stick divergence is conceptually to the trace distance, one might expect it to have a similar behaviour when it comes to expansion coefficients. Nevertheless, also here, for the depolarizing channel the coefficient is zero. The reasoning, however, is different to the previous examples. 

An important property of many divergences is faithfulness. However, the Hockey-Stick divergence for $\gamma>1$ does not have this property. That means, for $\gamma>1$, 
\begin{align}
    E_{\gamma}(\rho\|\sigma) \notimplies \rho=\sigma. 
\end{align}
We use again the example states $\rho$ and $\sigma$ from Equation~\eqref{Eq:rs-examples}. It can be checked that
\begin{align}
    E_{\gamma}(\rho\|\sigma)=(1-\gamma(1-\epsilon))_+,
\end{align}
and hence 
\begin{align}
    E_{\gamma}(\rho\|\sigma)=0 \iff \epsilon\leq\frac{\gamma-1}{\gamma}. 
\end{align}
Similarly, 
\begin{align}
    E_{\gamma}(\cD_p(\rho)\|\cD_p(\sigma))=\left(1-\frac{p}2-\gamma\left[(1-p)(1-\epsilon)-\frac{p}2\right]\right)_+,
\end{align}
and hence 
\begin{align}
    E_{\gamma}(\cD_p(\rho)\|\cD_p(\sigma))=0 \iff \epsilon\leq\frac{\gamma-1}{\gamma}\frac{1-\frac{p}2}{1-p}. 
\end{align}
This allows us to choose for any $p>0$ an $\epsilon^\star$ such that
\begin{align}
    \frac{\gamma-1}{\gamma} < \epsilon^\star < \frac{\gamma-1}{\gamma}\frac{1-\frac{p}2}{1-p}, 
\end{align}
for which 
\begin{align}
    E_{\gamma}(\cD_p(\rho)\|\cD_p(\sigma))=0, \quad\text{but}\quad E_{\gamma}(\rho\|\sigma)>0. 
\end{align}
The states $\rho$ and $\sigma$ with the parameter $\epsilon^\star$ are then an example of states that achieve
\begin{align}
    \widecheck\eta_\gamma(\cD_p)=0. 
\end{align}
Of course, for $\gamma=1$ the Hockey-Stick divergence is the trace distance which is faithful and therefor this strategy does not apply. 

It is intriguing that most expansion coefficients seem to become trivial for the qubit depolarizing channel, but not the one using the trace distance. Exploring the usefulness of expansion coefficients further seems like an interesting problem. 

\section{Expansion of the generalized amplitude damping channel}\label{App:GAD-exp}

In this section, we make precise and prove one statement from the main text. We summarize it in the following proposition. 
\begin{proposition}
    For the GAD channel $\cA_{p,\eta}$ with $p\in[0,1]$ and $\eta\in[0,1]$, we have
    \begin{align}
        \widecheck\eta(\cA_{p,\eta}) = 1-\widecheck\alpha^H(\cA_{p,\eta}) = \eta. 
    \end{align}
\end{proposition}
\begin{proof}
 We always have $\widecheck\eta(\cA_{p,\eta}) \geq 1-\widecheck\alpha^H(\cA_{p,\eta})$, therefore it suffices to show upper and lower bound of $\eta$ on the respective quantities. 

 We start by showing $\widecheck\eta(\cA_{p,\eta})\leq \eta$. Here, we consider the test states $|0\rangle\langle0|$ and $|1\rangle\langle1|$. By direct calculation, we find that 
 \begin{align}
     \cA_{p,\eta}(|0\rangle\langle0|)-\cA_{p,\eta}(|1\rangle\langle1|) = \eta|0\rangle\langle0|-\eta |1\rangle\langle1|. 
 \end{align}
 Hence, $E_1(\cA_{p,\eta}(|0\rangle\langle0|)\|\cA_{p,\eta}(|1\rangle\langle1|))=\eta$. Since we chose orthogonal test states, we immediately have the desired upper bound. 

 For the second inequality, we need to show that there exists a degrading map $\cD$ and a generalized depolarizing channel $\cD_{q,\sigma}$ that give a feasible solution for $\widecheck\alpha^H(\cA_{p,\eta})$ and imply $\widecheck\alpha^H(\cA_{p,\eta})\leq1-\eta$. An explicit construction for these is given below in Lemma~\ref{Lem:GAD-deg-dep}. This concludes the proof. 
\end{proof}
We now prove the auxiliary lemma used in previous result. 
\begin{lemma}\label{Lem:GAD-deg-dep}
    Given the GAD channel $\cA_{p,\eta}$, there exists a dephasing channel $\cD^Z_b$ and a state $\sigma$, such that 
    \begin{align}
        \cD^Z_b\circ\cA_{p,\eta}=\cD_{1-\eta,\sigma}, 
    \end{align}
    where $\cD_{q,\sigma}=(1-q)\id+q\sigma$. Explicitly, the state $\sigma$ is given by, 
    \begin{align}
        \sigma=\begin{pmatrix}
p & 0 \\ 0 & 1-p \end{pmatrix},
    \end{align}
    and $b=\frac{1-\sqrt{\eta}}{2}$. 
\end{lemma}
\begin{proof}
Given the explicit from of the channel $\cD$ and the state $\sigma$, this can be directly calculated on the level of Choi matrices. For completeness, we state some intermediate steps here. We have the Choi matrices, 
\begin{align}
    J(\cA_{p,\eta}) &= \frac12\begin{pmatrix}
p+(1-p)\eta & 0 & 0 & \sqrt{\eta} \\ 0 & (1-p)(1-\eta) & 0 & 0 \\ 0 & 0 & p(1-\eta) & 0 
\\ \sqrt{\eta} & 0 & 0 & p\eta+(1-p) \end{pmatrix}, \\[5mm]
    J(\cD_{1-\eta,\sigma}) &= \frac12\begin{pmatrix}
p+(1-p)\eta & 0 & 0 & \eta \\ 0 & (1-p)(1-\eta) & 0 & 0 \\ 0 & 0 & p(1-\eta) & 0 \\ \eta & 0 & 0 & p\eta+(1-p) \end{pmatrix}.
\end{align}
They are only different in the two non-zero off-diagonal elements, which are then adjusted by the appropriate dephasing channel $\cD^Z_b$. 
\end{proof}

\end{document}